\documentclass[11pt]{article}
\usepackage{url}
\usepackage[utf8]{inputenc}
\usepackage[small]{caption}
\usepackage{graphicx}
\usepackage{booktabs}

\usepackage{amsmath,multirow}
\usepackage{amsthm,thm-restate,thmtools}
\usepackage{enumerate}
\usepackage{etoolbox}
\usepackage{subcaption}
\usepackage{xcolor,xfrac}
\usepackage[sort,numbers]{natbib}
\usepackage[english]{babel}
\usepackage{charter}
\usepackage{authblk}
\usepackage{framed}
\usepackage{tikz}
\usetikzlibrary{shapes.arrows,patterns,shapes.geometric,positioning,arrows.meta}
\usepackage{amssymb}
\usepackage{amsfonts}
\usepackage[T1]{fontenc}
\usepackage[top=1 in,bottom=1in, left=1 in, right=1 in]{geometry}
\usepackage{enumitem}
\usepackage{tabularx}
\usepackage{algorithm}
\usepackage{algorithmic}

\usepackage{xspace}
\usepackage[most]{tcolorbox}

\usepackage{mathtools}

\newcommand{\R}{\mathbb{R}}
\newcommand{\N}{\mathbb{N}}

\renewcommand{\emptyset}{\varnothing}

\newtheorem{theorem}{Theorem}
\newtheorem{lemma}{Lemma}
\newtheorem{proposition}{Proposition}

\newtheorem{definition}{Definition}

\newtheorem{example}{Example}
\newtheorem{remark}{Remark}

\usepackage{pifont}
\usepackage{nicefrac}

\newcommand{\cmark}{\checkmark}
\newcommand{\xmark}{$-$}

\definecolor{ownresult}{RGB}{0,70,140}

\usepackage{mdframed}
\mdfsetup{
    linewidth=1pt
}

\definecolor{inkC}{HTML}{333333}
\definecolor{univC}{HTML}{777777}
\definecolor{stayC}{HTML}{B3B3F2}   \definecolor{stayLine}{HTML}{8C8CE8}
\definecolor{eaC}{HTML}{A3DCA3}     \definecolor{eaLine}{HTML}{6FC46F}
\definecolor{sfC}{HTML}{F2A2A2}     \definecolor{sfLine}{HTML}{E57F7F}
\definecolor{eqC}{HTML}{E8A04A}
\definecolor{wuC}{HTML}{B5307E}
\definecolor{fourC}{HTML}{2A9D8F}
\definecolor{hlC}{HTML}{F2C14E}
\tikzset{
  rule/.style={circle, draw=inkC, fill=inkC, inner sep=1.3pt},
  ruleexp/.style={circle, draw=inkC, fill=white, line width=0.7pt, inner sep=1.3pt},
  ours/.style={circle, draw=white, fill=inkC, line width=0.4pt, inner sep=1.5pt},
  oursexp/.style={circle, draw=inkC, fill=white, line width=0.7pt, inner sep=1.3pt},
  oursstar/.style={star, star points=5, star point ratio=2.3, draw=inkC, fill=hlC, line width=0.5pt, inner sep=1.3pt},
  lab/.style={inner sep=1pt, font=\footnotesize, text=inkC},
  ourlab/.style={inner sep=1pt, font=\footnotesize\bfseries, text=inkC},
  call/.style={-{Stealth[length=4pt]}, line width=0.5pt},
  region/.style={line width=0.9pt},
}

\newcommand{\DMC}{\textsf{DMC}\xspace}
\newcommand{\SV}{\textsf{SV}\xspace}
\newcommand{\RFC}{\textsf{RFC}\xspace}
\newcommand{\eRFC}{\textsf{eRFC}\xspace}
\newcommand{\MES}{\textsf{MES}\xspace}
\newcommand{\NDMES}{\textsf{NDMES}\xspace}
\newcommand{\CMES}{\textsf{CMES}\xspace}
\newcommand{\ULMES}{\textsf{ULMES}\xspace}
\newcommand{\RULMS}{\textsf{R-ULMS}\xspace}
\newcommand{\WULMS}{\textsf{W-ULMS}\xspace}
\newcommand{\LULMS}{\textsf{L-ULMS}\xspace}
\newcommand{\DULMES}{\textsf{eULMES}\xspace}
\newcommand{\WVS}{\textsf{WVS}\xspace}
\newcommand{\EVS}{\textsf{EVS}\xspace}

\newcommand{\ShVa}{\operatorname{SV}}
\newcommand{\MC}{\operatorname{MC}}

\newcommand{\STAY}{\textsc{Stay}\xspace}
\newcommand{\SSTAY}{\textsc{S-Stay}\xspace}
\newcommand{\EA}{\textsc{EA}\xspace}
\newcommand{\PART}{\textsc{Part}\xspace}

\newcommand{\OD}{\textsc{OD}\xspace}
\newcommand{\SF}{\textsc{SF}\xspace}
\newcommand{\IR}{\textsc{IR}\xspace}

\newcommand{\CR}{\operatorname{CR}}
\newcommand{\ND}{\operatorname{ND}}
\newcommand{\supp}{\operatorname{supp}}
\newcommand{\CES}{\textsf{CES}\xspace}

\usepackage{hyperref}
\hypersetup{
    colorlinks=true,
    linkcolor=blue,
    citecolor=green!60!black,
    urlcolor=blue
}
\usepackage[capitalise,noabbrev]{cleveref}

\title{Participation Incentives in Online Cooperative Games\thanks{This manuscript is the extended journal version of our AAMAS 2026 paper~\citep{AGS26a}.}}

\begin{document}
	
	\author[1]{Haris Aziz}
	\author[1]{Yuhang Guo}
	\author[2,3]{Zhaohong Sun}
	\affil[1]{University of New South Wales}
	\affil[2]{Kyushu University}
	\affil[3]{CyberAgent}
	\date{}
	
	\maketitle
	
\begin{abstract}
We study cooperative games in which the coalition forms online and the value it generates must be distributed irrevocably among the players present at each time step. We study participation incentives through four axioms: no player gains by delaying their arrival, every newly arriving player who adds value is rewarded immediately, every earlier player whose presence is essential receives a strictly positive share of the new increment, and a player who is a dummy in the current coalition has received nothing. We show that an equal-sharing rule can satisfy all four axioms but requires exponentially many value queries. This barrier is unavoidable as we prove that no equal-sharing rule satisfying the four axioms can be computed exactly with polynomially many value queries. In contrast, weighted sharing rules resolve the issue. For a family of rules that share each increment among a sharing set selected by a greedy deletion scan in reverse arrival order, we characterize exactly the weights under which the four axioms hold. Every rule in this family needs only $O(n^2)$ value queries, and among the desired rules, a unique one has weights that Lorenz-dominate those of all others. We finally characterize a more general class of weighted sharing rules, termed ordinal weighted sharing rules, whose weights depend only on the arrival order and on comparisons between coalition values, that satisfy the four axioms.
\end{abstract}

\section{Introduction}\label{intro}
Cooperative game theory studies how self-interested players share the benefits of cooperation when binding agreements among groups are feasible. Canonical models assume a static grand coalition and ask how its total value should be allocated; prominent solution concepts include the Shapley value~\citep{SHAP53a}, the Banzhaf index~\citep{Banz64a}, and the nucleolus~\citep{Schm69a}.
In many real-world contexts, however, coalitions do not emerge instantaneously, and value must be paid out before the final coalition is known. Consider the formation of a startup: the venture may begin with a handful of founders, while engineers, designers, and advisers join over time, each bringing distinct skills and resources. Deferring all compensation until the final composition is known is typically infeasible, and it may not even be clear when the set of contributors has stabilized. Collaborative machine learning offers a related motivation. In federated learning, data owners participate in successive training rounds, and their contributions can be valued as the model improves; the federated Shapley value of \citet{WRZ+20a}, for instance, values the participants of each round by the gain they bring to the model trained in earlier rounds, so the value of a participant depends on the round in which she participates. This round-based valuation does not itself prescribe immediate payments or one-time entry into a growing coalition. Mining pools in cryptocurrencies are a third example. Miners contribute computing power to a shared pool, and each block reward is divided among them when it is found. Under proportional rewards, miners can earn more than their fair share by mining for a pool only early in a round, a strategy known as pool hopping~\citep{Rose11a}, or by delaying the report of a found block solution~\citep{SBBR16a}.

These scenarios motivate our model of \emph{online cooperative games}, in which each arriving player may increase the coalition's value and that increment must be irrevocably allocated among the players present. The examples above suggest what a payment rule should achieve. A newcomer whose arrival makes a product launch possible, or whose data improves the model, expects to be paid when she contributes. Founders and early data owners on whom the current value still depends expect a share of later gains, not only of the gains they created themselves. A participant who adds nothing, such as an adviser with no influence on the product or a data owner whose data duplicates data already present, should receive nothing. Finally, no participant should gain by delaying her arrival, a timing concern analogous to pool hopping under proportional rewards.

The formal model of online cooperative games with strategic arrivals was introduced by \citet{GZZ+24a,ZGT+26a}. They study the \emph{Incentive for Early Arrival} (\EA), the \emph{Incentive to Stay} (\STAY), and ex-ante fairness, termed \textit{Shapley-Fairness} (\SF), which requires each player's expected payoff under uniformly random arrival orders to equal her Shapley value. They show that applying the Shapley value at every stage can reduce a player's accrued payoff, and that giving each newly arrived player her marginal contribution can reward delay. They propose the \textit{Rewarding First Critical Player} (\RFC) rule and its extension \eRFC to general games, and show that \STAY, \EA, and \SF cannot all be satisfied in general; on 0--1 games, \RFC satisfies all three whenever this is feasible. They also observe that, because \RFC awards each unit of value to the first critical player, later participants may lose their incentive to join, and they leave open whether mechanisms can trade some early-arrival incentive for such long-term participation incentives~\citep[Section~6.1]{ZGT+26a}. Like them, we focus on a single growing coalition and on how its value is distributed.

In this paper, we turn these requirements into axioms. Besides \EA and \STAY, we introduce three new axioms. \textit{Incentive for Participation} (\PART) requires that a newly arriving player who adds positive value receive a positive payment immediately. \textit{Strong Incentive to Stay} (\SSTAY) refines \STAY: whenever an arrival adds value, every earlier player who is critical for the current coalition value must receive a strictly positive part of the increment. This is a benefit-sharing requirement, and it applies even when contributions are additive. \textit{Online Dummy} (\OD) requires that a player who is a dummy in the current coalition has received nothing so far. Finally, we consider classical \textit{Individual Rationality} (\IR), which guarantees each player at least her stand-alone value. As we show, \IR can conflict with \SSTAY, because the entrant's stand-alone value may absorb all of the new value. This paper addresses the following questions:

\begin{quote}
	\textit{For the online cooperative game setting, what are the key participation incentive properties? How do the existing rules fare with respect to these properties? Can we design new rules that perform better with respect to these properties, and at what computational cost?}
\end{quote}

\subsection{Our Contribution}
We study participation incentives in online cooperative games with strategic arrivals. Beyond \STAY and \EA, studied by~\citet{GZZ+24a,ZGT+26a}, we introduce the axioms \SSTAY, \PART, and \OD. Our results are as follows.
\begin{enumerate}
	\item[(1)] We formulate the participation incentive axioms and show that each benchmark rule, namely \DMC, \SV, and \eRFC~\citep{GZZ+24a,ZGT+26a} and the Weighted Value-Sharing mechanisms of~\citet{GZZ26a}, violates at least one of them (\Cref{prop:limit_existing_rules}).
	\item[(2)] We study equal-sharing rules, which divide each increment equally among a sharing set. Sharing among all players present satisfies all four axioms except \OD, sharing among the critical players satisfies all except \EA, and sharing among the non-dummy players satisfies all four but needs exponentially many value queries (\Cref{prop:MES_NDMES_axiom} and \Cref{prop:ndmes-query-lower}). We show that no equal-sharing rule satisfying the four axioms can be computed exactly with polynomially many value queries (\Cref{thm:no_existence_equal_share}). We also propose the Upward Lexicographic Marginal Equal Sharing (\ULMES) rule, which selects an inclusion-minimal sharing set by a greedy deletion scan in reverse arrival order and satisfies \PART, \SSTAY, and \OD with $O(n^2)$ value queries (\Cref{prop:ulmes_properties}). It fails \EA in general but satisfies it on 0--1 games, and its extension \DULMES through the GM decomposition of \citet{GZZ+24a,ZGT+26a} satisfies all four axioms at an exponential query cost (\Cref{app:eulmes}). 
	\item[(3)] We keep the sharing set of \ULMES and divide each increment unequally. A simple reweighting, the Reweighted Upward Lexicographic Marginal Sharing (\RULMS) rule, satisfies all four axioms with $O(n^2)$ value queries (\Cref{prop:rulms_axioms}). For the family \WULMS of all such reweightings, we characterize exactly the weights that satisfy the four axioms (\Cref{thm:wulms}): earlier recipients receive weakly larger fractions, and the entrant's fraction, whenever she shares the increment with others, does not increase with the number of players present. The family contains a unique rule, \LULMS, whose weights Lorenz-dominate those of every other rule in the family satisfying the four axioms (\Cref{prop:lulms}).
	\item[(4)] We characterize the ordinal weighted sharing rules, whose weights depend only on the arrival order and on comparisons between coalition values, that satisfy all four axioms (\Cref{thm:characterization_weighted_rule}).
	\item[(5)] Finally, we study individual rationality (\IR) and show that no rule satisfies \IR in general and IR is incompatible with \SSTAY even on additive games. We show that zero-monotonicity is exactly the condition under which \STAY and \IR can be satisfied under every arrival order, under which we propose a simple refinement that first pays each arriving player her singleton coalition value and then applies a given rule to the remaining surplus. This preserves \STAY, \EA, \PART, \OD, and \SF, and further preserves \SSTAY at every arrival with strict surplus (Appendix~\ref{app:ir}).
\end{enumerate}
We summarize the results in \Cref{table:summary-online-value-sharing-rules} and \Cref{fig:complete_picture}. Proofs omitted from the main text are given in the appendix.

\begin{figure}[!htbp]
\centering
\begin{subfigure}[b]{0.49\textwidth}
\centering
\begin{tikzpicture}[x=0.55cm, y=0.55cm, remember picture]
  \draw[univC, line width=0.8pt, rounded corners=0.5cm] (-7,-4.4) rectangle (7,4.9);
  \coordinate (landA) at (6.1,-1.9);   
  \node[lab, anchor=north west, font=\scriptsize, text=univC] at (-6.8,4.75) {Online efficient rules};
  \def\regSTAY{(0,-1.9) ellipse [x radius=6.1, y radius=2.25]}
  \def\regSF{[rotate around={-32:(-2.3,0.9)}] (-2.3,0.9) ellipse [x radius=1.9, y radius=4.0]}
  \def\regEA{[rotate around={32:(2.3,0.9)}] (2.3,0.9) ellipse [x radius=1.9, y radius=4.0]}
  \fill[stayC, opacity=0.30] \regSTAY;  \draw[stayLine, region] \regSTAY;
  \fill[sfC, opacity=0.22] \regSF;      \draw[sfLine, region] \regSF;
  \fill[eaC, opacity=0.30] \regEA;      \draw[eaLine, region] \regEA;
  \node[lab, text=stayLine!80!black] at (0,-3.65) {\STAY};
  \node[lab, text=sfLine] at (-2.55,1.25) {\SF};
  \node[lab, text=eaLine!85!black] at (2.55,1.25) {\EA};
  \node[ruleexp] (sv) at (0,2.55) {};
  \node[lab] (svl) at (3.6,4.2) {\SV};  \draw[call, inkC] (svl.west) to[out=180, in=30] (sv);
  \node[rule] (erfc) at (-4.1,-1.45) {};  \node[lab, above=1pt] at (erfc) {\eRFC};
  \node[ours] (ces) at (-2.85,-0.75) {};  \node[ourlab, above=1pt] at (ces) {\CES};
  \node[rule] (dmc) at (-3.4,-2.3) {};    \node[lab, below=1pt] at (dmc) {\DMC};
  \node[ours] (ul) at (0,-1.55) {};       \node[ourlab, below=1pt] at (ul) {\ULMES};
  \node[ours] (rul) at (3.1,-0.75) {};    \node[ourlab, above=1pt] at (rul) {\RULMS};
  \node[oursstar] (lul) at (3.6,-2.0) {}; \node[ourlab, below=1pt] at (lul) {\LULMS};
\end{tikzpicture}
\caption{\emph{Scope under \EA, \STAY, and \SF}}
\end{subfigure}\hfill
\begin{subfigure}[b]{0.49\textwidth}
\centering
\begin{tikzpicture}[x=0.55cm, y=0.55cm, remember picture]
  \coordinate (landB) at (-7,-1.9);    
  \fill[stayC, opacity=0.30, rounded corners=0.5cm] (-7,-4.4) rectangle (7,4.9);
  \draw[stayLine, line width=0.8pt, rounded corners=0.5cm] (-7,-4.4) rectangle (7,4.9);
  \node[lab, anchor=south west, font=\scriptsize, text=stayLine!80!black] at (-6.8,-4.3) {\STAY};
  \def\regEQ{(0,-2.7) ellipse [x radius=5.2, y radius=1.05]}
  \def\regWU{[rotate around={54.372:(-2.15,0.3)}] (-2.15,0.3) ellipse [x radius=4.591, y radius=1.05]}
  \def\regFOUR{[rotate around={125.628:(2.15,0.3)}] (2.15,0.3) ellipse [x radius=4.591, y radius=1.05]}
  \fill[eqC, opacity=0.12] \regEQ;    \draw[eqC, region, dash pattern=on 3pt off 2pt] \regEQ;
  \fill[wuC, opacity=0.12] \regWU;    \draw[wuC, region] \regWU;
  \fill[fourC, opacity=0.16] \regFOUR; \draw[fourC, region] \regFOUR;
  \begin{scope} \clip \regFOUR; \fill[wuC, opacity=0.40] \regWU; \end{scope}
  \begin{scope} \clip \regFOUR; \fill[pattern=north east lines, pattern color=eqC] \regEQ; \end{scope}
  \node[lab, text=eqC, font=\scriptsize] at (0,-2.7) {Equal sharing};
  \node[lab, text=wuC, font=\scriptsize, rotate=54.372] at (-2.15,0.3) {\WULMS};
  \node[lab, text=fourC, font=\scriptsize, rotate=-54.372] at (3.41,1.2) {\PART, \SSTAY, \EA, \OD};
  \node[rule] (erfc) at (-5.4,2.9) {};    \node[lab, right=1pt] at (erfc) {\eRFC};
  \node[ours] (ces) at (-5.6,1.5) {};     \node[ourlab, right=1pt] at (ces) {\CES};
  \node[rule] (dmc) at (-4.45,-2.3) {};
  \node[lab, anchor=east] (dmcl) at (-4.85,-1.2) {\DMC};  \draw[inkC, line width=0.4pt] (dmcl.south east) -- (dmc);
  \node[ours] (ul) at (-3.75,-2.75) {};
  \node[ourlab, anchor=north] (ull) at (-3.0,-3.5) {\ULMES};  \draw[inkC, line width=0.4pt] (ul) -- (ull.north);
  \node[ours] (mes) at (2.45,-2.8) {};   \node[ourlab, below=1pt] at (mes) {\MES};
  \node[oursexp] (nd) at (4.35,-2.6) {};
  \node[ourlab, anchor=north east] (ndl) at (6.8,-3.5) {\NDMES};  \draw[inkC, line width=0.4pt] (nd) -- (ndl.north);
  \node[oursexp] (eul) at (2.9,-0.95) {};
  \node[ourlab, anchor=west] (eull) at (4.3,-1.35) {\DULMES};  \draw[inkC, line width=0.4pt] (eul) -- (eull.west);
  \node[ours] (rul) at (-0.3,3.05) {};
  \node[ourlab] (rull) at (-2.6,4.15) {\RULMS};  \draw[call, inkC] (rull.east) to[out=0, in=160] (rul);
  \node[oursstar] (lul) at (0.3,3.45) {};
  \node[ourlab] (lull) at (2.7,4.15) {\LULMS};   \draw[call, inkC] (lull.west) to[out=180, in=20] (lul);
\end{tikzpicture}
\caption{\emph{Participation axioms within \STAY.}}
\end{subfigure}
\begin{tikzpicture}[remember picture, overlay]
  \draw[stayLine!85!black, line width=0.9pt, -{Stealth[length=5pt, width=4pt]}, shorten <=1pt, shorten >=1pt]
    (landA) to[out=-25, in=205] (landB);
\end{tikzpicture}%
\caption{Landscape of online sharing rules. (a) Rules satisfying \EA (green), \SF (red), and \STAY (lavender); the \STAY rules are exactly the weighted sharing rules (\Cref{sec:weighted}), and \EA and \SF meet only outside \STAY. (b) Within \STAY: equal-sharing rules (dashed orange), \WULMS rules (magenta), and rules satisfying \PART, \SSTAY, \EA, and \OD (teal); no rule lies in all three. In the teal region, equal-sharing rules need $2^{\Omega(n)}$ value queries (\Cref{thm:no_existence_equal_share}), while \WULMS rules are characterized in \Cref{thm:wulms} and need $O(n^2)$. Filled (hollow) markers: polynomially (exponentially) many value queries; bold labels: our rules; star: \LULMS.}
\label{fig:complete_picture}
\end{figure}

\begin{table}[!htbp]
	\centering
	\small
	\setlength{\tabcolsep}{6pt}
	\renewcommand{\arraystretch}{1.15}
	\begin{tabular}{lcccccccc}
		\toprule
		& \multicolumn{3}{c}{Prior axioms} & \multicolumn{3}{c}{\textbf{New axioms}} & & Exact\\
		\cmidrule(lr){2-4}\cmidrule(lr){5-7}
		Rule & \EA & \STAY & \SF & {\PART} & {\SSTAY} & {\OD} & \IR & queries\\
		\midrule
		\multicolumn{9}{l}{\textit{Benchmark rules}}\\
		\DMC   & \xmark & \cmark & \cmark & \cmark & \xmark & \cmark & \cmark$^{\ddagger}$ & $O(n)$\\
		\SV    & \cmark & \xmark & \cmark & \cmark & \xmark & \cmark & \cmark$^{\ddagger}$ & exp.\\
		\eRFC  & \xmark & \cmark & \cmark & \xmark & \xmark & \cmark & \xmark & $O(n^2)$\\
		\addlinespace[2pt]
		\multicolumn{9}{l}{\textit{Equal-sharing rules and their extensions}}\\
		\textbf{\MES}    & \cmark & \cmark & \xmark & \cmark & \cmark & \xmark & \xmark & $O(n)$\\
		\textbf{\NDMES}  & \cmark & \cmark & \xmark & \cmark & \cmark & \cmark & \xmark & exp.\\
		\textbf{\CMES}   & \xmark & \cmark & \xmark & \cmark & \cmark & \cmark & \xmark & $O(n^2)$\\
		\textbf{\ULMES}  & \xmark$^{\dagger}$ & \cmark & \xmark & \cmark & \cmark & \cmark & \xmark & $O(n^2)$\\
		\textbf{\DULMES} & \cmark & \cmark & \xmark & \cmark & \cmark & \cmark & \xmark & exp.\\
		\textbf{\CES}    & \xmark & \cmark & \cmark & \cmark & \cmark & \cmark & \xmark & $O(n^2)$\\
		\addlinespace[2pt]
		\multicolumn{9}{l}{\textit{Weighted sharing rules}}\\
		\textbf{\RULMS}  & \cmark & \cmark & \xmark & \cmark & \cmark & \cmark & \xmark & $O(n^2)$\\
		\textbf{\LULMS}  & \cmark & \cmark & \xmark & \cmark & \cmark & \cmark & \xmark & $O(n^2)$\\
		\addlinespace[2pt]
		$\boldsymbol{\rho}^{\IR}$ $^{\ddagger}$ & \multicolumn{6}{c}{\footnotesize as $\rho$, but \SSTAY only at strict-surplus prefixes} & \cmark & as $\rho$\\
		\bottomrule
	\end{tabular}
	\caption{Summary of guarantees (\cmark: satisfied, \xmark: violated). The last column gives the number of value queries for exact deterministic computation; ``exp.'' means that exponentially many are needed in the worst case. $^{\dagger}$\EA holds on 0--1 games (\Cref{prop:ulmes_properties}). $^{\ddagger}$Only on zero-monotone games (\Cref{prop:ir-impossible}). For rules with \STAY and \EA, \SF fails by \citet{GZZ+24a}; for \CMES and \ULMES, see \Cref{app:cmes}. Bold: our axioms and rules. Benchmark entries for \EA, \STAY, \SF, and queries are from \citet{GZZ+24a,ZGT+26a}.}
	\label{table:summary-online-value-sharing-rules}
\end{table}

\subsection{Related Work}\label{sec:related_work}
\paragraph{Cooperative games.}
Cooperative game theory~\citep{vNM44a} studies settings in which players benefit from forming coalitions and making collective decisions. A central question is how the value created by a coalition should be distributed among its members, and solution concepts are typically justified by axioms such as efficiency, symmetry, stability, or consistency. The Shapley value~\citep{SHAP53a} initiated the axiomatic study of allocation and underlies a large subsequent literature. \citet{GILL59a} introduced the core, \citet{SHUB59a} studied market games, \citet{AUMA64a} introduced the bargaining set, and \citet{Schm69a} introduced the nucleolus. \citet{ROSO92a} survey two-sided matching, which connects cooperative solution concepts with practical matching markets. A further line of work studies coalition formation with hedonic preferences, in which players care only about the members of their own coalition~\citep{ALRE04a,BOW08a,AZBR12a,AZSA16a}, as well as the closely related group activity selection problem~\citep{DEK+12a}. Classical and computational treatments of cooperative game theory can be found in several books (see, e.g.,~\citep{BDT08a,CEW22a}). These models take the set of players as given when value is allocated. We instead study a coalition that grows over time, where each allocation is irrevocable and must be made without knowledge of future arrivals.

\paragraph{Online cooperative games with strategic arrivals.}
The closest work is the online value-sharing model of \citet{GZZ+24a,ZGT+26a}. Their \RFC and \eRFC rules prioritize the first critical player, and their impossibility result explains why \STAY, \EA, and \SF cannot all be imposed on general games. They also note that \RFC weakens incentives for later participants~\citep[Section~6.1]{ZGT+26a}; \PART and \SSTAY formalize this concern. \citet{ZGT+26a} describe our conference version~\citep{AGS26a} as dropping \SF in favor of participation axioms; the present version also studies a Shapley-fair rule with these axioms. We ask which positive participation guarantees can be retained together with \EA and \OD, characterize weighted sharing rules, and study exact value-query complexity. The distinction between a large explicit decomposition and a hard allocation problem is essential: \citet[Algorithm~2 and Theorem~7]{ZGT+26a} compute \eRFC without constructing the GM decomposition. Our lower bounds for \NDMES and \DULMES therefore exhibit games that agree on all but one coalition yet require different sharings, rather than infer hardness from the size of a representation. \citet{Gu26a} studies alternative decompositions for extending \RFC; we retain the GM decomposition and analyze its induced allocations.

\citet{GZZ26a} study fairness within individual arrival orders on simple games. They discuss equal sharing among critical players and prove its failure of \EA, including on solvable simple games. Their Egalitarian Value-Sharing mechanism satisfies \SF, \STAY, and \EA on solvable simple games~\citep[Theorem~2]{GZZ26a}. Moreover, its formula assigns a positive share to every critical player at the pivotal prefix~\citep[Mechanism~2]{GZZ26a}; in our terminology, this implies \PART and \SSTAY on that domain. Our \MES, \NDMES, and \DULMES guarantees cover all monotone games without requiring \SF. For \CES, our contributions are its \SF and participation analysis on general monotone games and the direct equal-share payment formula and query bound; equal sharing among critical players is also illustrated in \citet[Section~6.1]{ZGT+26a}. The threshold representation and deletion queries used in that computation are inherited from \citet[Section~5.3]{ZGT+26a}. The weights in \citet{GZZ26a} depend on the arrival rank among the critical players and are defined only on SOCGs, whereas the weights characterized here are attached to the \ULMES sharing set or determined by ordinal comparisons and apply to all monotone games. \citet{ZZG+24a} study the corresponding problem for cost sharing under online strategic arrivals. They propose Shapley-fair shuffle cost-sharing mechanisms and show that, unlike in value sharing, every monotone cost function admits a mechanism satisfying online individual rationality (cost shares never increase as players arrive), incentives for early arrival, and Shapley fairness. Their payments divide a cost rather than distribute coalition value.

\paragraph{Online coalition formation.}
Another line of work studies cooperative games in an online manner with a different goal: arriving players must be placed into coalitions, and the aim is to approximate social welfare or to reach stable outcomes. \citet{FMM+21a} study the online coalition structure generation problem, in which arriving agents are irrevocably assigned to coalitions, and analyze competitive social welfare. \citet{BURO23a} study online coalition formation in hedonic games under random arrival or coalition dissolution, and \citet{BuRoR24} study stability in online coalition formation. In the model of \citet{ZSS+25a}, each arriving player can create a new coalition or join an existing one, and value-sharing policies are designed to optimize the competitive ratio with respect to social welfare. Unlike the strategic-arrival model, this literature typically assumes that agents' preferences are revealed truthfully and does not consider manipulation of arrival times. Online and dynamic perspectives have also been applied to matching and coalition formation more broadly, including dynamically stable matching~\citep{Laur22} and dynamics reaching individually stable coalition structures~\citep{BBWa23a}. We instead take the growing coalition as given and allocate its value subject to participation axioms. Moreover, arrival orders here are arbitrary: random orders enter only the averaging requirement of \SF, and the four participation and dummy-player guarantees are required for every order.

\paragraph{Online mechanism design.}
In the strategic-arrival model, a player's arrival time is private information that she may manipulate, so the problem can be viewed as one of dynamic mechanism design. Mechanism design in dynamic environments concerns agents with private information, where the goal is to elicit that information while making decisions without knowledge of future events~\citep{PARK07a,BeVa19a}; representative mechanisms include the dynamic pivot mechanism~\citep{BeVa10a} and the efficient dynamic mechanism of \citet{AtSe13a}. \citet{LANI00a} initiated the competitive analysis of incentive-compatible online auctions, and \citet{FRPA03a} introduced the term online mechanism design. \citet{PASI03a} and \citet{PYS04a} extend VCG-based mechanisms to online settings modeled as Markov decision processes. Online matching is a central topic in online algorithm design~\citep{KVV90a,KAPR00a,FMM09a,GKM+19a}, and sequential mechanism design has been studied for problems such as scheduling~\citep{PORT04a,CHL+16a}; broader treatments of online computation and online stochastic combinatorial optimization are given by \citet{BOEL05a} and \citet{HEBE06a}. Two of these models are closest in spirit to ours. \citet{KPS14a} study irrevocable allocations of multiple resources to arriving agents with Leontief preferences, with fairness and truthfulness concerning reported demands. Our resource is transferable coalition value, and the strategic comparison concerns arrival positions. \citet{CHL+16a} design schedules and payments for agents with private job characteristics, including release times, and analyze competitive welfare. We fix the valuation function and characterize payment-sharing rules rather than elicit job information or select a schedule. More generally, online mechanisms elicit reports and select outcomes, whereas the sharing rules studied here have no reports to elicit: the value function is given, and the strategic choice is when to join.

\section{Preliminaries}\label{preliminary}
	\subsection{Model}
	An online cooperative game (OCG) $G$ is a triple $(N,v,\pi)$, where $N$ is a finite set of players with $|N|=n$, $v: 2^N\to \mathbb{R}_{+}$ is the valuation function mapping a subset of players to a non-negative real number, and $\pi \in \Pi(N)$ is a permutation of $N$ representing the arrival order of all the players over all possible permutations $\Pi(N)$. Given any subset $S \subseteq N$, $S$ creates a coalition with value $v(S)$. In this paper, we focus on general valuation functions that are normalized, meaning $v(\emptyset)=0$, and monotone, meaning $v(S)\leq v(T)$ whenever $S\subseteq T\subseteq N$.
	A valuation function $v:2^N\to \mathbb{R}$ is \textit{submodular} if $v(S)+v(T)\geq v(S\cup T)+v(S\cap T)$ for all $S,T\subseteq N$, and \textit{supermodular} if the reverse inequality holds. It is \textit{subadditive} if $v(S)+v(T)\geq v(S\cup T)$ for all $S,T$, and \textit{superadditive} if $v(S)+v(T)\leq v(S\cup T)$ for all disjoint $S,T\subseteq N$. For singleton value, we will write $v(i)$ instead of $v(\{i\})$.  

	Given a permutation $\pi$, for any pair of players $i, j \in N$, let $i \prec_{\pi} j$ denote that player $i$ arrives earlier than player $j$, according to the permutation $\pi$. We write $i\preceq_\pi j$ if $i=j$ or $i\prec_\pi j$. An \emph{online value-sharing rule} $\phi$ maps the game $G=(N,v,\pi)$ to an allocation vector $\phi(G) = (\phi(G,i))_{i\in N}$, where $\phi(G, i)$ denotes the value assigned to player $i$. For any player $i\in N$, we assume $\phi(G, i) \geq 0$ and $\sum_{i \in N} \phi(G, i) = v(N)$.
	
	We now introduce two definitions that will serve as the foundation for the axioms and sharing rules. The first is the notion of a \textbf{prefix subgame}. A subset $S\subseteq N$ is called a \emph{prefix} of the arrival order $\pi$ if $S$ consists of the first $|S|$ players to arrive. In this case, we write $S \sqsubseteq \pi$. Given an OCG $G = (N, v, \pi)$ and a prefix $S \sqsubseteq \pi$, the prefix subgame is defined as $G^S = (S, v_{\mid S}, \pi_{\mid S})$, where $v_{\mid S}$ is the valuation function restricted to coalitions $C \subseteq S$, and $\pi_{\mid S}$ denotes the arrival order of players in $S$.

	For any player $i$ in the arrival order $\pi$, we denote by $N_{\pi \mid i}$ the prefix consisting of all players who arrive no later than $i$ (including $i$). The corresponding prefix subgame is then $G^{N_{\pi \mid i}}=(N_{\pi \mid i}, v_{\mid N_{\pi \mid i}}, \pi_{\mid N_{\pi \mid i}})$. When the context is clear, we write $G^i:=G^{N_{\pi\mid i}}$ for this prefix subgame. 
	We next define the \textbf{Simple Online Cooperative Game} (SOCG) whose valuation function takes values in $\{0,1\}$. Formally, an OCG $G=(N,v,\pi)$ is an SOCG if $v(S)\in \{0,1\}$ for all $S\subseteq N$ and $v(N)=1$. By monotonicity, there is a unique \emph{pivotal} player, upon whose arrival the value of the coalition jumps from $0$ to $1$ and the value remains $1$ afterwards.

	We measure computation in the value-oracle model, where each query returns the value of a coalition of already-arrived players. Polynomial query bounds are defined in terms of the number of players, rather than the size of an explicit table of coalition values. All query lower bounds concern deterministic algorithms that compute exact allocations. The valuation function is fixed and is not strategically reported. Players may strategically delay their arrival, with the relative order of all other players held fixed.
	
	\subsection{Existing Incentive Axioms}
	We revisit the axioms introduced by \citet{GZZ+24a} and further developed in~\citep{ZGT+26a}. The first property is termed \textit{Incentive to Stay},
	\footnote{
		In the original papers \citep{GZZ+24a,ZGT+26a}, this property is referred to as Online Individual Rationality (OIR). However, it differs from the classic notion of individual rationality. As discussed in Section~\ref{sec:Participation_Axioms}, we revisit the axiom of individual rationality (IR) that aligns with the classic notion.}
	which guarantees that each player's shared value is non-decreasing as more players arrive. This encourages arrived players to stay in the grand coalition for more potential rewards. 
	
	\begin{definition}[Incentive to Stay (\STAY)~\citep{GZZ+24a,ZGT+26a}]
		\label{def:STAY}
		An online value-sharing rule $\phi$ satisfies \emph{incentive to stay (\STAY}) if, for every OCG $G=(N,v,\pi)$, for any two prefix subgames $G^S=(S,v_{\mid S}, \pi_{\mid S})$ and $G^T=(T,v_{\mid T}, \pi_{\mid T})$, where $T \subseteq S$, and $T,S \sqsubseteq \pi$, every player $q \in T$ satisfies $\phi(G^T,q) \leq \phi(G^S,q)$.
	\end{definition}
	\STAY is the first natural participation incentive axiom studied by \citet{GZZ+24a}. Next, we revisit another participation incentive axiom called \textit{Incentive for Early Arrival}. Recall that in the online cooperative game model, the arrival time of each agent is treated as private information. So players might strategically choose to delay their arrival for extra benefits. The axiom of EA requires that for every player $i$, when fixing the arrival order of all other players, arriving as early as possible is a dominant strategy for player $i$.
	\begin{definition}[Incentive for Early Arrival (\EA)~\citep{GZZ+24a,ZGT+26a}]
		\label{def:EA}
		An online value-sharing rule $\phi$ satisfies \EA if, for any two OCGs $G = (N, v, \pi)$ and $G' = (N, v, \pi')$, for each player $i$, it always holds $\phi(G, i) \geq \phi(G', i)$ whenever $\pi_{\mid N \setminus \{i\}} = \pi'_{\mid N \setminus \{i\}}$ and $N_{\pi \mid i} \subset N_{\pi' \mid i}$. 
	\end{definition}

	To verify \EA, it suffices to consider delays by one position, in which a player swaps with the player arriving immediately after her. This reduction is used in~\citep{GZZ+24a,ZZG+24a,GZZ26a}; we state it below and prove it in Appendix~\ref{app:adjacent}.
	\begin{restatable}{lemma}{EAAdjacent}\label{lem:EA_adjacent}
		An online value-sharing rule $\phi$ satisfies \EA if and only if the following holds: for every OCG $G=(N,v,\pi)$ with $\pi=(\dots,i,j,\dots)$, letting $G'=(N,v,\pi')$ where $\pi'=(\dots,j,i,\dots)$, we have $\phi(G,i)\ge\phi(G',i)$.
	\end{restatable}
	Before introducing the next axiom called \emph{Shapley-Fairness}, we first introduce the concept of marginal contribution and the classic Shapley Value. Given an OCG $G=(N,v,\pi)$, for each player $i\in N$, we define the \textit{marginal contribution} of $i$ to a coalition $S$ in $G$ as $\MC(G,S,i) = v(S\cup \{i\}) - v(S)$. Based on the definition of marginal contribution, we introduce the Shapley Value. 
	\begin{definition}[Shapley Value (SV)~\citep{SHAP53a}]
		\label{def:SV}
		Given an OCG $G=(N,v,\pi)$, each player $i$'s Shapley Value is
		\begin{align*}
			\ShVa(G,i)=\frac{1}{|N|!}\sum_{S\subseteq N\setminus \{i\}} |S|! \cdot (|N|-|S|-1)! \cdot \MC(G,S,i).
		\end{align*}
	\end{definition}
	
	The Shapley value assigns to each player her expected marginal contribution under a uniformly random player order. It ensures that players are rewarded fairly based on how much they add to the value of any coalition they join. The following axiom uses the Shapley value as an ex-ante benchmark for online sharing.
	
	\begin{definition}[Shapley-Fairness (SF)~\citep{GZZ+24a,ZGT+26a}]
		\label{def:SF}
		Given an OCG $G=(N,v,\pi)$, an online value-sharing rule $\phi$ is said to satisfy \emph{Shapley-Fairness (\SF)} if, for each player $i\in N$, 
		\begin{align*}
			\frac{1}{|N|!}\sum_{\sigma \in \Pi(N)} \phi((N,v,\sigma),i) = \ShVa(G,i).
		\end{align*}
	\end{definition}
	
	Intuitively, \SF requires that, for any OCG, if all arrival orders of the players are equally likely, then the expected payoff of each player coincides with her Shapley value. However, Shapley-Fairness only guarantees fairness in expectation across permutations; for a fixed arrival order, the payoffs assigned by an SF rule need not match the Shapley values exactly. In online settings, by contrast, only one arrival order is realized, and the total value must be allocated sequentially as players appear. The incompatibility of \SF, \STAY, and \EA on general games, established by \citet{GZZ+24a}, motivates studying attainable combinations of fairness and participation requirements. 
	
	\subsection{Participation Incentive Axioms}
	\label{sec:Participation_Axioms}
	In the previous section, we revisited two participation incentive axioms, \textsc{Stay} and \textsc{EA}, originally introduced by \citet{GZZ+24a}. We now turn to additional natural axioms that further refine our understanding of participation incentives in online cooperative games. To formulate these axioms, we first introduce an auxiliary notion for players whose arrival generates a strictly positive marginal contribution, which we refer to as \emph{incremental players}.
	\begin{definition}[Incremental Player]
		\label{def:Contributional_Player}
		Given an OCG $G=(N,v,\pi)$, for each player $i$ in the arrival order $\pi$, if $v(N_{\pi \mid i}) > v(N_{\pi \mid i}\setminus \{i\})$, then player $i$ is called an incremental player under permutation $\pi$ in $G$.
	\end{definition}

	We also introduce notation for two sets of players that will be used throughout the analysis. For any coalition $S\subseteq N$, let $\CR(S)=\{j\in S:v(S)>v(S\setminus\{j\})\}$ be the set of critical players in the game restricted to $S$: removing any such player reduces the coalition's value. Similarly, let $\ND(S)=\{j\in S:\text{there is }T\subseteq S\setminus\{j\} \text{ with }v(T\cup\{j\})>v(T)\}$ be the set of non-dummy players in that restricted game. Every critical player is non-dummy, and a player who is non-dummy in $S$ remains non-dummy in every larger coalition. In particular, an arriving player $i$ is incremental exactly when $i\in\CR(N_{\pi\mid i})$.
	
	Building on this notion, we introduce two participation incentive axioms. The first, \textit{Strong Incentive to Stay} (\SSTAY), strengthens \STAY by requiring positive increments at relevant arrivals.
	
	\begin{definition}[Strong Incentive to Stay (\SSTAY)]
		\label{def:S-STAY}
		An online value-sharing rule $\phi$ satisfies \textit{Strong Incentive to Stay} (\SSTAY) if, for any OCG $G = (N, v, \pi)$, it satisfies the \STAY axiom, and, for every incremental player $i$, and every player $j \prec_{\pi} i$ such that $v(N_{\pi \mid i}) > v(N_{\pi \mid i} \setminus \{j\})$, it holds that for player $j$, $\phi(G^{i}, j) > \phi(G^{i-}, j)$, where $G^{i-}$ denotes the prefix subgame immediately preceding $G^i$, i.e., the one with player set $N_{\pi\mid i}\setminus\{i\}$.
	\end{definition}
	Intuitively, \SSTAY requires not only that each player's cumulative allocation be non-decreasing, but also that whenever a newly arrived player $i$ creates positive value, every earlier player $j$ who is critical for the current coalition value, in the sense that removing $j$ would lower it, receives a share of this newly created value. Such a player need not be necessary for the marginal contribution of $i$.

	We refer to the second axiom as \textit{Incentive for Participation} (\PART), which ensures that each arriving player receives an immediate share of value if they contribute to the coalition. We consider this a minimal and natural requirement in the online setting. 
	
	\begin{definition}[Incentive for Participation (\PART)]
		An online value-sharing rule $\phi$ satisfies \textit{Incentive for Participation} (\PART) if, for any OCG $G = (N, v, \pi)$ and every player $i$, whenever $i$ is an incremental player, it holds that $\phi(G^{i}, i) > 0$.
	\end{definition}
	
	Intuitively, \PART requires that in any OCG $G$, every incremental player $i$ receives a strictly positive share immediately upon joining the coalition. The two axioms address different recipients of a positive increment. In a two-player SOCG where only the pair has value~1, \SSTAY guarantees a positive increment to the first player when the second arrives, while \PART guarantees a positive payment to the entrant. Neither requirement alone ensures that a payment meets the recipient's stand-alone benchmark. Apart from these two new axioms, we also consider the classical notion of \textit{Individual Rationality} (\IR), which aligns with its definition in standard cooperative games and can likewise be viewed as a participation incentive axiom.
	\begin{definition}[Individual Rationality (\IR)~\citep{AdKr17a}]
		Given an OCG $G=(N,v,\pi)$, an online value-sharing rule $\phi$ is individually rational (\IR) if for each player $i\in N$, $\phi(G, i) \geq v(\{i\})$.
	\end{definition}
	
	Although we do not impose Shapley-Fairness on our main rules, we also consider the Online Dummy axiom (\OD), which excludes null players from the allocation. The Online Dummy axiom is defined within each prefix subgame $G^i$ and requires that every \emph{dummy player} in $G^i$ has received nothing up to $G^i$. It imposes the classical null-player requirement in every prefix subgame. We retain the term \emph{dummy} for null players throughout the paper; it does not refer to the more general classical dummy condition of a constant, possibly positive, marginal contribution. We begin by formally defining a \emph{dummy player}.
	
	\begin{definition}[Dummy Player~\citep{AdKr17a}]
		Given an OCG $G=(N,v,\pi)$, for player $i$, if $v(S\cup\{i\})=v(S)$ for every $S\subseteq N$, then player $i$ is a dummy player in $G$. 
	\end{definition}
	
	In classical cooperative games, the \textit{null-player} axiom requires that every null player receives zero value. We now extend this axiom to the online setting by introducing the \textit{Online Dummy} axiom. 
	
	\begin{definition}[Online Dummy (OD)]
		An online value-sharing rule $\phi$ satisfies Online Dummy (\OD) if for each prefix subgame $G^{i}$ and every dummy player $j$ in game $G^i$, $\phi(G^{i},j)=0$.
	\end{definition}
	The OD axiom requires that dummy players receive zero payoff not only in the final allocation but also in every prefix subgame. In particular, any dummy player in a certain prefix subgame receives no share at that stage, even though they may later contribute positively through interactions with subsequently arriving players. The next section revisits benchmark sharing rules and shows the limitations of \DMC, \SV, \eRFC and the Weighted Value-Sharing (\WVS) mechanisms with respect to the aforementioned axioms.
	
	\section{Limitations of Existing Rules} \label{limitation_of_existing_rules}
	We now revisit the online value-sharing rules of \citet{GZZ+24a,ZGT+26a} and \citet{GZZ26a} and check which participation incentive axioms they satisfy.

	\medskip
	\textbf{Distribute Marginal Contribution (\DMC) rule}. \DMC gives each arriving player her whole marginal contribution: for every OCG $G=(N,v,\pi)$ and every player $i$, $\phi(G,i)=\MC(G^i, N_{\pi\mid i}\setminus \{i\},i)$. \citet{GZZ+24a,ZGT+26a} prove that \DMC satisfies \STAY and \SF but fails \EA in general, and that it satisfies \EA if and only if the valuation function is submodular.
	\medskip

	\textbf{Shapley Value (\SV) rule}. \SV computes the Shapley value of each prefix subgame and pays it to the players present. After the last arrival, each player's payoff is her Shapley value in $G$. \citet{GZZ+24a} prove that \SV satisfies \EA and \SF but fails \STAY.
	\medskip

	\textbf{Rewarding First Critical Player (\RFC) rule}. \citet{GZZ+24a} propose \RFC for SOCGs and extend it to general OCGs through the \emph{Greedy Monotone} (GM) decomposition. In an SOCG with pivotal player $i$, \RFC gives the whole unit to the earliest-arriving player of $\CR(N_{\pi\mid i})$, that is, to the earliest player whose removal makes the current coalition losing. For general OCGs, \citet{GZZ+24a,ZGT+26a} decompose the game into SOCGs with positive coefficients (Appendix~\ref{app:gm}), apply \RFC to each of them, and add up the results with these coefficients; the resulting rule is \eRFC. It satisfies \STAY and \SF but not \EA, even on SOCGs. Following \citet{GZZ+24a}, we call an SOCG \emph{unsolvable} if some player $j$ with $v(j)=0$ is the unique critical player of some winning coalition and \emph{solvable} otherwise. \citet{GZZ+24a} shows no rule satisfies \STAY, \EA, and \SF on all SOCGs and that \RFC satisfies all three on every solvable SOCG.
	\medskip

	\textbf{Weighted Value-Sharing (\WVS) mechanisms}. \citet{GZZ26a} propose a family of mechanisms for SOCGs that pay only critical players. Consider an SOCG $G=(N,v,\pi)$ with pivotal player $i$, let $C=\CR(N_{\pi\mid i})$ and $m=|C|$, let $P$ be the shortest prefix of $\pi$ with $v(P\cup\{i\})=1$, and let $m'=|\CR(P\cup\{i\})|$. Each mechanism is specified by a nonincreasing weight function $w:\N_{+}\to\R_{+}$ with $w(1)>0$. If $m=1$, player $i$ receives the unit. If $m\geq 2$, player $i$ receives $\alpha=w(m')/\sum_{k=1}^{m'}w(k)$, and the remaining $1-\alpha$ is divided among $C\setminus\{i\}$, where the player with arrival rank $r$ in $C\setminus\{i\}$ receives the fraction $w(r)/\sum_{k=1}^{m-1}w(k)$ of $1-\alpha$. All payments are made at the arrival of $i$. The \emph{Egalitarian Value-Sharing} (\EVS) mechanism uses constant weights, so that player $i$ receives $\frac{1}{m'}$ and each $j\in C\setminus\{i\}$ receives $\frac{1}{m-1}(1-\frac{1}{m'})$. \citet{GZZ26a} show that these mechanisms satisfy \SF, \STAY, and \EA on every solvable SOCG.
	\medskip

	Since \EA, \STAY, and \SF are incompatible in general, we relax \SF, focus on the participation incentive axioms \STAY, \EA, \SSTAY, \PART, and \OD, and ask whether all of them can be satisfied simultaneously. Unfortunately, none of the rules above satisfies all of them.
	\begin{proposition}
		\label{prop:limit_existing_rules}
		The following statements hold.
		\begin{enumerate}
			\item \SV and \DMC satisfy \PART, whereas \eRFC does not.
			\item \DMC, \SV, and \eRFC do not satisfy \SSTAY.
			\item \DMC and \eRFC fail \EA in general.
			\item On SOCGs, \EVS satisfies \PART, \SSTAY, and \OD, whereas every \WVS mechanism fails \EA.
		\end{enumerate}
	\end{proposition}

	\begin{proof}
		Under \DMC, an incremental player $i$ immediately receives $\MC(G^i,N_{\pi\mid i}\setminus\{i\},i)>0$. Under \SV, $\phi(G^i,i)$ is the Shapley value of $i$ in $G^i$, an average of marginal contributions that are nonnegative by monotonicity, one of which (for the orders of $N_{\pi\mid i}$ ending with $i$) equals $\MC(G^i,N_{\pi\mid i}\setminus\{i\},i)>0$; hence both rules satisfy \PART. For \eRFC, consider the SOCG with $N=\{1,2\}$, $\pi=(1,2)$, $v(\{1\})=v(\{2\})=0$, and $v(\{1,2\})=1$. \eRFC gives the whole value to player~$1$, the first critical player, so the incremental player~$2$ receives nothing, which violates \PART.

		Since \SV fails \STAY, it fails \SSTAY. For \DMC and \eRFC, consider the SOCG with $N=\{1,2,3\}$, $\pi=(1,2,3)$, $v(\{1,2,3\})=1$, and $v(S)=0$ for every other $S$. \DMC gives the unit to player~$3$, whose arrival creates it, although players $1$ and $2$ are critical in $\{1,2,3\}$. \eRFC gives the unit to player~$1$, so the critical player~$2$ receives nothing. Both violate \SSTAY.

		For the third statement, let $N=\{1,2\}$, $v(N)=1$, and $v(S)=0$ for every other $S$. Under \DMC, player~$1$ receives $0$ under $(1,2)$ and $1$ under $(2,1)$, so she gains by delaying. \eRFC fails \EA on some unsolvable SOCGs~\citep{GZZ+24a,ZGT+26a}.

		For the last statement, consider an SOCG with pivotal player $i$ and use the notation above. Every winning subset of $N_{\pi\mid i}$ contains $C$, because removing a member of $C$ from $N_{\pi\mid i}$ leaves a losing coalition. Since $P\cup\{i\}\subseteq N_{\pi\mid i}$ is winning, it contains $C$, and every member of $C$ is critical in it. Hence $m'\geq m$. The only incremental arrival is that of $i$, who receives $1$ or $\frac{1}{m'}>0$, which gives \PART. If an arrived player $j$ is critical in $N_{\pi\mid i}$, then $m\geq 2$, and EVS gives $j$ the positive fraction $\frac{1}{m-1}(1-\frac{1}{m'})$, since $m'\geq m\geq 2$; this gives \SSTAY. A player who is a dummy in a prefix subgame containing $N_{\pi\mid i}$ is a dummy in $N_{\pi\mid i}$, so she does not belong to $C$ and receives nothing, and no payment is made before the arrival of $i$; this gives \OD. Finally, let $N=\{1,2,3\}$ and let $v(S)=1$ if $\{1,2\}\subseteq S$ or $\{1,3\}\subseteq S$, and $v(S)=0$ otherwise. Under $\pi=(2,1,3)$, the pivotal player is $1$, $C=\{1,2\}$, and $P=\{2\}$, so $m=m'=2$, and every \WVS mechanism gives player~$1$ the fraction $\frac{w(2)}{w(1)+w(2)}\leq\frac{1}{2}$, since $w$ is nonincreasing. Under $\pi'=(2,3,1)$, in which player~$1$ delays, the pivotal player is again $1$, but $C=\CR(N)=\{1\}$, so player~$1$ receives the whole unit. Hence every \WVS mechanism fails \EA. This game is unsolvable, because $v(\{1\})=0$ and player~$1$ is the unique critical player of $N$, so this failure is consistent with \citet[Theorem~2]{GZZ26a}.
	\end{proof}

	\DMC and \eRFC fail both \EA and \SSTAY, and \eRFC also fails \PART. \SV satisfies \EA only because its final allocation does not depend on the arrival order, and it fails \STAY. The \WVS mechanisms are defined only on SOCGs and fail \EA on unsolvable SOCGs. In the next section, we show that equal-sharing rules can satisfy all of the participation incentive axioms at once.

\section{Equal-Sharing Rules}\label{sec:equal-sharing}
Both \DMC and \RFC assign each positive marginal contribution to a single player: \DMC rewards the entrant, whereas \RFC rewards the earliest-arriving critical player. However, \PART and \SSTAY may require positive payments to different players from the same increment. We therefore begin with a simple equal-sharing scheme that divides each positive marginal contribution equally among all players who have arrived. We call this the \emph{Marginal Equal Share} (\MES) rule. Every incremental entrant and every critical existing player receives a positive payment, so \PART and \SSTAY follow immediately. As we show below, \MES also satisfies \EA on all normalized monotone games.

Sharing among all players present can, however, reward current dummies and thereby reduce the value available to contributors. To exclude these players, the \emph{Non-Dummy Marginal Equal Share} (\NDMES) rule divides each positive increment equally among the current non-dummy players. This restriction retains the participation guarantees of \MES while also satisfying \OD. Thus \NDMES satisfies all four axioms: \PART, \SSTAY, \EA, and \OD. Its main limitation is computational: identifying all current non-dummy players, and even computing the resulting allocation exactly, requires exponentially many value queries in the worst case. This raises the question of whether an equal-sharing rule can satisfy all four axioms using only polynomially many value queries.

We now formally define the equal-sharing scheme. Given an OCG $G=(N,v,\pi)$ and an arriving player $i$, write $\MC_i=\MC(G^i,N_{\pi\mid i}\setminus\{i\},i)=v(N_{\pi\mid i})-v(N_{\pi\mid i}\setminus\{i\})$. An \textbf{equal-sharing rule} selects a nonempty sharing set $S^i\subseteq N_{\pi\mid i}$ whenever $\MC_i>0$ and divides $\MC_i$ equally among its members, so each player $j\in S^i$ receives an additional $\MC_i/|S^i|$ at the arrival of $i$. Nothing is paid when $\MC_i=0$. We begin with the following three equal-sharing rules, defined for each prefix subgame $G^i$ with $\MC_i>0$.
\begin{enumerate}
	\item \textbf{Marginal Equal Share (\MES) rule}: the increment is shared among all current players, $S^i=N_{\pi\mid i}$. 
	\item \textbf{Non-Dummy Marginal Equal Share (\NDMES) rule}: the increment is shared among the current non-dummy players, $S^i=\ND(N_{\pi \mid i})$. 
	\item \textbf{Critical Marginal Equal Share (\CMES) rule}: the increment is shared among the current critical players, $S^i=\CR(N_{\pi \mid i})$. 
\end{enumerate}

\begin{proposition}\label{prop:mes-ndmes}\label{prop:MES_NDMES_axiom}
	The following statements hold for the \MES, \NDMES, and \CMES rules.
	\begin{enumerate}
		\item[(1)] \MES, \NDMES, and \CMES satisfy \PART and \SSTAY.
		\item[(2)] \NDMES and \CMES satisfy \OD, but \MES does not.
		\item[(3)] \MES and \NDMES satisfy \EA, but \CMES does not.  
		\item[(4)] \MES can be computed with $O(n)$ value queries, \CMES with $O(n^2)$ value queries, and \NDMES with $O(2^n)$ value queries. 
	\end{enumerate}
\end{proposition}
\begin{proof}
For every prefix subgame $G^i$, since a critical player is not a dummy player, we observe that the three sharing sets are nested $ \CR(N_{\pi\mid i})\subseteq\ND(N_{\pi\mid i})\subseteq N_{\pi\mid i}$. All three rules make only nonnegative payments, so accrued payoffs never decrease and all three satisfy \STAY. 

\textbf{PART and S-STAY.} Let $\MC_i>0$. Then $v(N_{\pi\mid i}\setminus\{i\})<v(N_{\pi\mid i})$, so the incremental player $i$ is critical in $N_{\pi\mid i}$. By the inclusions above, $i$ and every critical existing player belong to the sharing set of each of the three rules, and each of them receives $\MC_i/|S^i| >0$ at $i$'s arrival. This shows that all the rules satisfy \PART and \SSTAY.

\textbf{\OD.} Let $j$ be a dummy in a prefix subgame $G^k$. For every $i$ with $j\in N_{\pi\mid i}\subseteq N_{\pi\mid k}$, every coalition of $N_{\pi\mid i}$ is a coalition of $N_{\pi\mid k}$, so $j$ is also a dummy in $G^i$. Hence $j\notin\ND(N_{\pi\mid i})$, and therefore $j\notin\CR(N_{\pi\mid i})$. Neither \NDMES nor \CMES has paid $j$ up to $k$, so $\phi(G^k,j)=0$. For \MES, let $N=\{a,b,d\}$ and consider the valuation with $v(S)=1$ if $\{a,b\}\subseteq S$ and $v(S)=0$ otherwise. Consider the arrival order $\pi=(d,a,b)$. Player $d$ is a dummy, but \MES pays her $\frac{1}{3}$ at $b$'s arrival.

\textbf{\EA.} By \Cref{lem:EA_adjacent}, it suffices to compare $\pi=(A,i,j,\dots)$ with $\pi'=(A,j,i,\dots)$, where $A$ is the set of earlier players, arriving in the same order in both, and all later players also arrive in the same order. For a set $T'$ of current players, write sharing set $S(T')=T'$ for \MES and $S(T')=\ND(T')$ for \NDMES. The sharing set depends only on the current players, not on their order, so the two orders produce the same payments up to the arrival of $A$ and after both $i$ and $j$ have arrived. It therefore suffices to compare the payments to $i$ at the arrivals of $i$ and $j$. Let $T=A\cup\{i,j\}$. Under $\pi'$, these payments to $i$ amount to $\frac{v(T)-v(A\cup\{j\})}{|S(T)|}$ if this increment is positive, and to zero otherwise. The latter case is trivial. In the former case, $i$ is not a dummy in $T$, so $i\in S(T)$. Under $\pi$, let $a$ be the payment to $i$ at her own arrival. If $v(A\cup\{i\})>v(A)$, then $a=\frac{v(A\cup\{i\})-v(A)}{|S(A\cup\{i\})|}$, and $\emptyset\neq S(A\cup\{i\})\subseteq S(T)$, since a non-dummy of $A\cup\{i\}$ is a non-dummy of $T$. If $v(A\cup\{i\})=v(A)$, then $a=0$. In both cases, $a\geq\frac{v(A\cup\{i\})-v(A)}{|S(T)|}$. At the arrival of $j$, player~$i$ receives $\frac{v(T)-v(A\cup\{i\})}{|S(T)|}$, which is zero if this increment is zero. Using monotonicity, her payments under $\pi$ satisfy
\begin{align*}
  a+\frac{v(T)-v(A\cup\{i\})}{|S(T)|}&\geq \frac{v(A\cup\{i\})-v(A)}{|S(T)|}+\frac{v(T)-v(A\cup\{i\})}{|S(T)|}\\
  &=\frac{v(T)-v(A)}{|S(T)|} \geq \frac{v(T)-v(A\cup\{j\})}{|S(T)|},
\end{align*}
which is the payment to $i$ under $\pi'$. This proves \EA for \MES and \NDMES. It remains to show that \CMES fails \EA. Let $N=\{a,b,c,d\}$ and $v(S)=1$ if $\{a,b,c\}\subseteq S$ or $\{a,b,d\}\subseteq S$, and $v(S)=0$ otherwise~\citep[Example~3]{GZZ26a}. Under $(c,a,b,d)$, the value becomes $1$ at $b$'s arrival with $\CR(\{a,b,c\})=\{a,b,c\}$, so $b$ receives $\frac{1}{3} $. Under $(c,a,d,b)$, which delays $b$ by one position, the value becomes $1$ at $b$'s arrival with $\CR(\{a,b,c,d\})=\{a,b\}$, because removing $c$ or $d$ leaves a winning coalition. Player $b$ then receives $\frac{1}{2}>\frac{1}{3} $, which shows \CMES fails \EA.

\textbf{Query Complexity.} \MES needs only the value of each prefix, that is, $n$ queries. At an incremental arrival, \CMES additionally queries $v(N_{\pi\mid i}\setminus\{j\})$ for each current player $j$, which gives $O(n^2)$ queries in total. For \NDMES, the values of all coalitions of $N_{\pi\mid i}$ determine which current players are non-dummies; storing the answers, all arrivals together query each coalition of $N$ at most once, which gives at most $2^n$ queries.
\end{proof}

We note that the simplicity of \NDMES as a formula does not make its sharing sets easy to find, because whether a player is a dummy can depend on her contribution to any coalition of the current prefix. Exact deterministic computation of \NDMES requires exponentially many value queries in the worst case; Appendix~\ref{app:ndmes-computation} gives this lower bound and positive results for structured valuations. This motivates selecting a smaller sharing set whose members can be certified to be non-dummies with few queries. 

We therefore study the \emph{Upward Lexicographic Marginal Equal Share} (\ULMES) rule, which starts from all current players and inspects them from the latest to the earliest arrival. It removes a player whenever the remaining coalition still generates the current value, and the players left at the end share the increment equally. We show that \ULMES satisfies \PART, \SSTAY, and \OD with polynomially many value queries, and that it satisfies \EA on SOCGs but not in general. We then show that no equal-sharing rule can do better: no equal-sharing rule satisfying \PART, \SSTAY, \EA, and \OD can be computed with polynomially many value queries.

\paragraph{Upward Lexicographic Marginal Equal Share (\ULMES) rule.}
Consider a prefix subgame $G^i$ with $\MC_i>0$, and list the current players in arrival order as $N_{\pi\mid i}=\{j_1,\dots,j_t\}$, so that $j_t=i$. \ULMES starts from all current players and inspects them from the latest to the earliest, removing a player whenever the remaining coalition still generates the current value. Formally, let $S_t=N_{\pi\mid i}$ and, for $r=t,t-1,\dots,1$, define 
\begin{align*}
  S_{r-1}=
  \begin{cases}
    S_r\setminus\{j_r\}, & \text{if } v(S_r\setminus\{j_r\})=v(N_{\pi\mid i}),\\[2pt]
    S_r, & \text{otherwise}.
  \end{cases}
\end{align*}
The sharing set of \ULMES is $S^i=S_0$, and each player in $S^i$ receives $\MC_i/|S^i|$. Every removal keeps the value at $v(N_{\pi\mid i})$, so $v(S^i)=v(N_{\pi\mid i})$. The procedure uses only the prefix subgame $G^i$ and makes one value query per current player. We show below that $S^i$ is an inclusion-minimal coalition generating the current value $v(N_{\pi\mid i})$. Since the procedure examines late arrivals first and removes a player whenever the remaining players still generate the current value, it keeps earlier arrivals in the sharing set in preference to later ones, which explains the name of the rule. This preference concerns only the selection of the sharing set and does not yield \EA. \Cref{ex:ulmes-not-ea} below shows that \ULMES fails \EA in general. 

\begin{proposition}
	\label{prop:ulmes_properties}
	The \ULMES rule satisfies \PART, \SSTAY, and \OD, and it satisfies \EA on SOCGs. It can be computed with $O(n^2)$ value queries. 
\end{proposition}

\begin{proof}
	We first show that $\CR(N_{\pi\mid i})\subseteq S^i\subseteq\ND(N_{\pi\mid i})$. Suppose some $S\subsetneq S^i$ satisfies $v(S)=v(N_{\pi\mid i})$, and pick $j_r\in S^i\setminus S$. When $j_r$ was inspected, the current coalition $S_r$ contained $S^i$, so $S\subseteq S_r\setminus\{j_r\}$, and monotonicity gives $v(S_r\setminus\{j_r\})=v(N_{\pi\mid i})$. The procedure would have removed $j_r$, a contradiction. Hence $S^i$ is inclusion-minimal. If some $j\in\CR(N_{\pi\mid i})$ were not in $S^i$, then $S^i\subseteq N_{\pi\mid i}\setminus\{j\}$ and $v(N_{\pi\mid i})=v(S^i)\leq v(N_{\pi\mid i}\setminus\{j\})<v(N_{\pi\mid i})$, which is impossible. If some $j\in S^i$ were a dummy in $G^i$, then $v(S^i\setminus\{j\})=v(S^i)$, contradicting inclusion-minimality. Since $\MC_i>0$ means $i\in\CR(N_{\pi\mid i})$, the incremental player and every critical existing player belong to $S^i$ and receive $\MC_i/|S^i|>0$, which gives \PART and \SSTAY. A player who is a dummy in $G^k$ is a dummy in every prefix subgame $G^i$ with $i\preceq_\pi k$ that contains her, so she belongs to no sharing set up to $k$; this gives \OD. At each arrival, \ULMES queries $v(N_{\pi\mid i})$ once, and the deletion scan makes one value query per current player, so the total over all $n$ arrivals is at most $\sum_{t=1}^{n}(t+1)=O(n^2)$ value queries.

	We finally prove that \ULMES satisfies \EA on SOCGs. By \Cref{lem:EA_adjacent}, it suffices to compare $G=(N,v,\pi)$, where $\pi=(\dots,i,j,\dots)$, with $G'=(N,v,\pi')$, where $\pi'$ swaps $i$ and $j$. In an SOCG, the only positive increment occurs at the first winning prefix, denoted by $U$, so a player's payoff is $1/|S|$ if she belongs to the \ULMES sharing set $S$ at $U$ and $0$ otherwise. The set $S$ is also obtained by scanning all of $N$ from the latest to the earliest arrival and removing each player whose removal leaves a winning coalition: every player arriving after $U$ is removed, since the coalition still contains $U$, and the scan then continues on $U$ exactly as \ULMES does. Run this scan under $\pi$ and under $\pi'$, and let $S$ and $S'$ be the results. The players arriving after both $i$ and $j$ are scanned identically, so both scans reach the pair with the same coalition $C$. Under $\pi$ the scan inspects $j$ before $i$, and under $\pi'$ it inspects $i$ before $j$. Suppose $i\in S'$. Then $i$ was kept under $\pi'$, so $C\setminus\{i\}$ is losing. In both scans, $j$ is removed exactly when $C\setminus\{j\}$ is winning. Under $\pi$, removing $i$ would leave $C\setminus\{i\}$ or $C\setminus\{i,j\}$, both losing by monotonicity, so $i$ is kept as well. Hence both scans leave the pair with the same coalition and treat the earlier players identically, so $S=S'$ and $\phi(G',i)=\phi(G,i)$. If $i\notin S'$, then $\phi(G',i)=0\le\phi(G,i)$. 
\end{proof}

However, \ULMES fails \EA in general. \Cref{ex:ulmes-not-ea} shows that a player 
may gain extra benefits by delaying her arrival.

\begin{example}\label{ex:ulmes-not-ea}
    Let $N=\{1,2,3,4\}$ and $0<x<y$. For every $T\subseteq N$, set the coalition valuations as follows.
	\begin{align*}
	v(T)=
	\begin{cases}
		x & \text{ if } T \supseteq \{1,2,3\} \text{ and } 4\notin T \\
		y & \text{ if } T \supseteq \{3, 4\} \\
		0 & \text{ otherwise }
	\end{cases}
	\end{align*}
    Under $\pi=(1,2,3,4)$, by \ULMES, player~$3$ first shares an increment $x$ with players~$1$ and~$2$. When player~$4$ arrives, the selected set becomes $\{3,4\}$ and shares $(y-x)$. If player~$3$ delays to $\pi'=(1,2,4,3)$, only her arrival creates value, and the selected set is $\{3,4\}$. Hence we have 
	\begin{align*}
		\phi(G,3)=\frac{x}{3}+\frac{y-x}{2} =\frac{y}{2}-\frac{x}{6} <\frac{y}{2}=\phi(G',3).
	\end{align*}
	By delaying, player~$3$ no longer shares the first $x$ units with players~$1$ and~$2$ and instead she shares all $y$ units with player~$4$ alone. Delaying therefore increases her share of the first $x$ units from $\frac{1}{3}$ to $\frac{1}{2}$. Hence \ULMES fails \EA in general.
\end{example}

\begin{remark}
	Although \ULMES fails \EA in general, it satisfies \EA on SOCGs. In an SOCG, only the pivotal prefix creates a positive increment. The upward lexicographic selection favors earlier players in this single allocation: if a player receives a positive share after delaying, arriving earlier gives her at least the same share. General valuations allow several positive increments with different sharing sets. A player may then benefit by delaying until a larger value is available and sharing it among fewer recipients, as in Example~\ref{ex:ulmes-not-ea}. One way to recover \EA is to apply \ULMES separately to each SOCG of the GM decomposition of \citet{GZZ+24a,ZGT+26a} and to aggregate the results. The resulting rule \DULMES satisfies \EA, \PART, \SSTAY, and \OD, but computing it exactly needs exponentially many value queries (Appendix~\ref{app:eulmes}).
\end{remark}

We next show that no equal-sharing rule satisfying all four axioms can be computed exactly with a polynomial number of value queries.

\begin{theorem}
	\label{thm:no_existence_equal_share}
	No equal-sharing rule that satisfies \PART, \SSTAY, \EA, and \OD in general can be computed exactly by a deterministic algorithm using a polynomial number of value queries.
\end{theorem}

\begin{proof}
	We begin by constructing a base OCG $G_0=(N,v_0,\pi_0)$. Let $m=\lfloor \frac{n-2}{3}\rfloor\geq 1$, so that $n-m\geq 2m+2$. We partition $N$ into $\{i\}$, a set $W$ of $2m$ players, and a set $D$ of $n-2m-1$ dummy players. Let $W_0=W\cup \{i\}$ and $k=|W_0|=2m+1$, and fix a number $x$ with $1-\frac{1}{k^2} < x < 1$. For every $S\subseteq N$, let
	\begin{align*}
		v_0(S)=
		\begin{cases}
			1 & \text{if } W_0\subseteq S, \\
			x & \text{if } i\in S,\ |S\cap W| > m, \text{ and } W_0 \not\subseteq S,\\
			0 & \text{otherwise,}
		\end{cases}
	\end{align*}
	and let $\pi_0=(D,W,i)$, with a fixed internal order of $D$ and of $W$. The value of a coalition depends only on whether it contains $i$ and on how many members of $W$ it contains, so every player of $D$ is a dummy in $v_0$. We further define a family of OCGs which are copies of $G_0$ except for a single coalition. Recall that $W$ has $2m$ players. We consider any possible subset $A$ of $W$ with size $|A|=m$. Let $W_A=\{i\} \cup A \cup D$. We define a new valuation function $v_A$ where for any $S\subseteq N$, $v_A(S)=x$ if $S= W_A$, otherwise $v_0(S)$. Put differently, the only difference between the valuation function $v_A$ and $v_0$ is the value of coalition $W_A$. We further observe that for any proper subset of $W_A$, the coalition value is $0$ while for any proper superset of $W_A$, the coalition value is at least $x$ as it contains $i$ and at least $m+1$ players from $W$. It is not hard to see that any game in this family is normalized and monotone. 

	We further define two other arrival orders with respect to each set $A$: $\pi_A=(A,D,W\setminus A, i)$ and $\pi'_A=(A,D,i,W\setminus A)$, in which the members of $A$ and those of $W\setminus A$ arrive in the internal order of $W$. Denote by $G_A=(N,v_A,\pi_A)$ and $G'_A=(N,v_A,\pi'_A)$. We further define $G_0^A=(N,v_A,\pi_0)$. 

	Consider any equal-sharing rule $\phi$ satisfying \PART, \SSTAY, \OD, and \EA. We focus on
	the value share received by player $i$ in different combinations of valuation functions and orders. For $G_0=(N,v_0,\pi_0)$, we claim that $\phi(G_0,i)=\frac{1}{k}$. To see this, before $i$ arrives, every prefix misses $i$ and has value of $0$, by definition of equal sharing rules, every existing player receives zero. When $i$ arrives, the prefix is $N$ and the value rises from $0$ to $1$. Removing any member of $W_0$ from $N$ leaves a coalition that does not contain $W_0$ and therefore has value at most $x<1$. So every member of $W_0$ is critical at $N$ and, by \PART and \SSTAY, every player in $W_0$ will be included in the sharing set. Moreover, by \OD, no player in $D$ belongs to the sharing set, since every player in $D$ is a dummy in $G_0$. Thus, the sharing set is exactly $W_0$. Hence, $\phi(G_0,i)=\frac{1}{k}$. 

	We next claim that for the OCG $G_0^A=(N,v_A,\pi_0)$, $\phi(G_0^A,i) < \frac{1}{k}$. To this end, we will show that $\phi(G_0^A,i) \leq \phi(G_A,i) \leq \phi(G_A',i) < \frac{1}{k}$. Firstly, for $G_A'=(N,v_A,\pi_A')$, the prefixes before $i$ miss $i$ and have value of $0$. When $i$ arrives, the prefix is exactly $W_A$ and the value rises from $0$ to $x$. Every proper subset of $W_A$ has value $0$, so every player in $W_A$ is critical at $W_A$, and therefore by \PART and \SSTAY, the increment $x$ is split equally among all of $W_A$. Note that $|W_A|=m+1+(n-(2m+1))=n-m\geq k+1$. Player $i$ receives at most $\frac{x}{k+1}$ at this arrival. Later arrivals raise the value by $1-x$ in total, we then obtain an upper bound of $\phi(G_A',i)$
	\begin{align}
		\label{eq:less_than_frac1k}
		\phi(G_A',i) \leq \frac{x}{k+1}+1-x=1-\frac{k}{k+1}x < 1 - \frac{k}{k+1}(1-\frac{1}{k^2})=\frac{1}{k}.
	\end{align}
	We next show that $\phi(G_A,i)\leq \phi(G_A',i)$ as $\pi_A$ can be viewed as the delayed order of $\pi_A'$ where $i$ delays her arrival. By \EA, the inequality $\phi(G_A,i)\leq \phi(G_A',i)$ holds. We finally prove that $\phi(G_0^A,i)\leq \phi(G_A,i)$. Recall the orders $\pi_0=(D,W,i)$ and $\pi_A=(A,D,W\setminus A, i)$. They differ in the positions of players other than $i$. Consider any arrival order $\sigma$ in which $i$ arrives last, and the game $G^A_\sigma=(N,v_A,\sigma)$. We know that before $i$ joins, every prefix has value of $0$ and every player receives $0$. When $i$ arrives, the value rises from $0$ to $1$, and removing $i$ or any player $j\in A$ from $N$ leaves a coalition that does not contain $W_0$ and therefore has value at most $x < 1$. By \PART and \SSTAY, $i$ and each $j\in A$ are in the sharing set so they will receive the same share. That is, for every $j\in A$, $\phi(G^A_\sigma,j)=\phi(G^A_\sigma,i)$. Now we consider a transition process from order $\pi_0$ to order $\pi_A$. List $A=\{j_1,\dots,j_m\}$ in the internal order of $W$. Starting from $\pi_0$, move the players $j_m,j_{m-1},\dots,j_1$ to the front, one at a time; it results in the order $\pi_A$. Each move advances one player, preserves the relative order of all other players, and keeps $i$ last. By \EA, the moved player's payment does not decrease, and by the equality $\phi(G^A_\sigma,j)=\phi(G^A_\sigma,i)$ for any order $\sigma$ with last arriving $i$, neither does the payment of $i$. This proves that $\phi(G_0^A,i)\leq \phi(G_A,i)$, and together with \eqref{eq:less_than_frac1k} we obtain $\phi(G_0^A,i)<\frac{1}{k}$. 

	Now consider any equal-sharing rule $\phi$ that satisfies \PART, \SSTAY, \EA, and \OD. When running $\phi$ on $G_0=(N,v_0,\pi_0)$, by our first claim, $\phi(G_0,i)=\frac{1}{k}$. Suppose that the rule never queries some $W_A$. Since $v_A$ agrees with $v_0$ on every other coalition, when running $\phi$ on $G_0^A=(N,v_A,\pi_0)$, it makes the same queries, receives the same answers, and reports the same payment $\frac{1}{k}$, contradicting our second claim that $\phi(G_0^A,i) < \frac{1}{k}$. Therefore, the run on $v_0$ queries every $W_A$. Distinct sets $A$ give distinct coalitions $W_A$, so the algorithm makes at least $\binom{2m}{m}=2^{\Omega(n)}$ queries.
\end{proof}

\section{Weighted Sharing Rules}\label{sec:weighted}
By \Cref{thm:no_existence_equal_share}, every equal-sharing rule that satisfies \PART, \SSTAY, \OD, and \EA needs $2^{\Omega(n)}$ value queries. We now ask whether some value-sharing rule satisfies all four axioms with polynomially many value queries. We answer this question positively with \emph{weighted sharing rules}, which may divide each positive increment unequally among the players.

A weighted sharing rule assigns a weight to each current player at every incremental arrival, and each player receives the corresponding fraction of the increment. Formally, a \emph{weight function} assigns to every OCG $G=(N,v,\pi)$ whose last arrival is incremental a vector $\lambda(G)=(\lambda_j(G))_{j\in N}$ with $\lambda_j(G)\geq 0$ and $\sum_{j\in N}\lambda_j(G)=1$. The \emph{weighted sharing rule} induced by $\lambda$ gives each player $j$ the fraction $\lambda_j(G^i)$ of $\MC_i$ at every incremental arrival $i$, and pays nothing at non-incremental arrivals. Thus, at every prefix subgame $G^k$ and for each existing player $j$,
\begin{align}\label{eq:weighted}
  \phi(G^k,j)=\sum_{\substack{i:\;j\preceq_\pi i\preceq_\pi k\\ \MC_i>0}}\lambda_j(G^i)\,\MC_i .
\end{align}
Every weighted sharing rule is efficient in each prefix and satisfies \STAY, because it distributes each increment in full and never revokes a payment. Conversely, every online rule that is efficient in each prefix and satisfies \STAY is a weighted sharing rule. At an incremental arrival $i$, divide each player's payoff increase by $\MC_i$, taking the entrant's previous payoff to be zero; by \STAY and efficiency, these ratios are nonnegative and sum to one. At a non-incremental arrival, efficiency and \STAY force every payoff increase to be zero. Equal sharing among $S^i$ is the special case $\lambda_j(G^i)=\frac{1}{|S^i|}$ for $j\in S^i$ and $\lambda_j(G^i)=0$ otherwise.

\subsection{Reweighted Upward Lexicographic Marginal Sharing (\RULMS) Rules}
We begin with a simple reweighting of \ULMES that satisfies all four axioms with polynomially many value queries. Consider a prefix subgame $G^i$ with $\MC_i>0$ and the \ULMES sharing set $S^i$. Instead of dividing $\MC_i$ equally among the players in $S^i$, we divide it unequally. Let $t=|N_{\pi\mid i}|$ and $s=|S^i|$. Each player in $S^i$ first receives the fraction $\frac{1}{t}$ of $\MC_i$, and the remaining fraction $1-\frac{s}{t}$ goes to the earliest-arriving player in $S^i$. Thus the earliest selected player receives the fraction $\frac{t-s+1}{t}$, and every other selected player, including the entrant, receives $\frac{1}{t}$. When $s=1$, the entrant is the only selected player and receives the whole increment. We call this rule \emph{Reweighted Upward Lexicographic Marginal Sharing} (\RULMS).

We revisit \Cref{ex:ulmes-not-ea} to explain how \RULMS removes the gain from delay under \ULMES. In the example, when player~$3$ arrives on time, the first $x$ units are created by the coalition $\{1,2,3\}$ and split three ways. If she delays until after player~$4$, all $y$ units are created at her arrival by the smaller coalition $\{3,4\}$, and equal sharing gives her half. The gain comes from the fact that, under equal sharing, the incremental player's fraction $\frac{1}{|S^i|}$ depends on the size of the sharing set, and she can shrink this set by delaying. Delay also has two effects that work against her: the number $t$ of players present at her arrival can only grow, and if she stays in the same sharing set, her rank in it can only move later. \RULMS gives the incremental player $\frac{1}{t}$, which can only fall when she delays, and gives the remainder to the earliest selected player, so a later rank in the same sharing set never earns more. In the example, player~$3$ receives $\frac{x}{3}+\frac{3(y-x)}{4}$ under $\pi=(1,2,3,4)$ and only $\frac{y}{4}$ under $\pi'=(1,2,4,3)$, so delay no longer pays.

\begin{proposition}
	\label{prop:rulms_axioms}
	The \RULMS rule satisfies \PART, \SSTAY, \EA, and \OD. Computing \RULMS requires $O(n^2)$ value queries.
\end{proposition}

We defer the proof to \Cref{sec:wulms}, where \RULMS is a special case of a family of weighted sharing rules. That section determines all ways of dividing $\MC_i$ among the players in $S^i$ that preserve the four axioms. 

\subsection{A Family of Upward Lexicographic Marginal Sharing Rules}\label{sec:wulms}
We now keep the sharing sets of \ULMES and vary only how each positive increment is divided among the selected players. This yields a family of rules, which we call \emph{weighted Upward Lexicographic Marginal Sharing} (\WULMS) rules. 

Consider an OCG $G=(N,v,\pi)$ and a prefix subgame $G^i$ with $\MC_i>0$. Let $t=|N_{\pi\mid i}|$ be the number of current players and $s=|S^i|$ the number of selected players. The \emph{rank} of a player in $S^i$ is her position in $\pi_{\mid S^i}$, with rank $1$ for the earliest selected player. By definition of \ULMES, the incremental player $i$ belongs to $S^i$, so she has rank $s$.

\begin{definition}[\WULMS rules]\label{def:wulms}
A \emph{weight family} is a collection $\lambda=(\lambda^{(t,s)})_{1\leq s\leq t}$ of vectors $\lambda^{(t,s)}=(\lambda^{(t,s)}_1,\dots,\lambda^{(t,s)}_s)$ with $\lambda^{(t,s)}_r\geq 0$ and $\sum_{r=1}^{s}\lambda^{(t,s)}_r=1$. The \WULMS rule induced by $\lambda$ operates as follows. At every prefix subgame $G^i$ with $\MC_i>0$, the player of rank $r$ in $S^i$ receives $\lambda^{(t,s)}_r \MC_i$, and every player outside $S^i$ receives nothing. 
\end{definition}

The same weight family is used in every OCG, so the shares at an incremental arrival depend only on the number of current players, the number of selected players, and the ranks. \ULMES is the \WULMS rule with $\lambda^{(t,s)}_r=\frac{1}{s}$ and \RULMS is the \WULMS rule with $\lambda^{(t,s)}_1=\frac{t-s+1}{t}$ and $\lambda^{(t,s)}_r=\frac{1}{t}$ for each $2 \leq r \leq s$. Every \WULMS rule uses $O(n^2)$ value queries. To see this, at each arrival, it queries the value of the current prefix and runs the deletion procedure of \ULMES, which makes one query per current player, and the weights require no further queries. In the notation of Section~\ref{sec:weighted}, a \WULMS rule is the weighted sharing rule that gives the rank-$r$ player of $S^i$ the weight $\lambda^{(t,s)}_r$ and every other player weight $0$.

We next give a complete characterization of rules in this family. 
\begin{theorem}\label{thm:wulms}
The \WULMS rule induced by $\lambda$ satisfies \EA, \PART, \SSTAY, and \OD if and only if
\begin{enumerate}
  \item[(a)] $\lambda^{(t,s)}_1 \geq \lambda^{(t,s)}_2\geq \cdots\geq \lambda^{(t,s)}_s>0$ for all $1\leq s \leq t$; and
  \item[(b)] $\lambda^{(t,s)}_s \leq \lambda^{(t',s')}_{s'}$ for all $t>t'\geq 1$, $2\leq s\leq t$, and $1\leq s'\leq t'$.
\end{enumerate}
\end{theorem}

Intuitively, Condition~(a) says that, within a sharing set, an earlier selected player never receives less than a later one, and the incremental player receives a positive fraction. Condition~(b) says that the incremental player's fraction does not increase with the number of current players, whatever the sizes of the sharing sets. 

\begin{proof}[Proof of Theorem~\ref{thm:wulms}]
Suppose conditions~(a) and (b) hold. By (a), every entry of every weight vector is positive. \PART holds immediately as the incremental player $i$ is in $S^i$, so $\phi(G^i,i)=\lambda^{(t,s)}_s\MC_i>0$. For \SSTAY, payments are nonnegative, so \STAY holds. If $j \prec_\pi i$ and $v(N_{\pi\mid i})>v(N_{\pi\mid i}\setminus\{j\})$, then $j\in \CR(N_{\pi\mid i})\subseteq S^i$ by the proof of \Cref{prop:ulmes_properties}, so $j$ receives a positive fraction of $\MC_i$ and $\phi(G^i,j)>\phi(G^{i-},j)$. For \OD, let $j$ be a dummy in $G^i$. Then $j$ is a dummy in every earlier prefix subgame containing her, so $j\notin \ND(N_{\pi\mid k})\supseteq S^k$ for every $k$ with $j \preceq_\pi k\preceq_\pi i$. Hence $\phi(G^i,j)=0$. For \EA, by \Cref{lem:EA_adjacent}, it suffices to compare $G=(N,v,\pi)$, where $\pi=(\dots,i,j,\dots)$, with $G'=(N,v,\pi')$, where $\pi'=(\dots,j,i,\dots)$. Let $A$ be the set of players preceding $i$ in $\pi$, set $T=A\cup\{i,j\}$, and let $t=|T|$. Player $i$ receives nothing before the pair arrives.

We first compare the deletion procedure under the two orders. Let $T'\supseteq T$ be a prefix of both orders, and let $S$ and $S'$ be the sets that the deletion procedure of \ULMES{} returns on $T'$ under $\pi$ and under $\pi'$, respectively. We show that if $i\in S'$, then $S=S'$ and the rank of $i$ in $S$ under $\pi$ is at most her rank in $S'$ under $\pi'$. Both procedures inspect the players arriving after the pair in the same order and with the same target value $v(T')$, so they reach the pair with the same coalition $C$. Under $\pi'$, player $i$ is inspected before $j$; since $i\in S'$, she is kept, so $v(C\setminus\{i\})<v(T')$. In both procedures, $j$ is removed exactly when $v(C\setminus\{j\})=v(T')$. Under $\pi$, removing $i$ would leave $C\setminus\{i\}$ or $C\setminus\{i,j\}$, whose values are below $v(T')$ by monotonicity, so $i$ is kept as well. Hence both procedures leave the pair with the same coalition and then treat the players of $A$ identically, so $S=S'$. Within this set, $i$ precedes $j$ under $\pi$ and follows $j$ under $\pi'$, and all other relative positions agree. Therefore the rank of $i$ under $\pi$ is at most her rank under $\pi'$.

\emph{Arrivals after the pair.} At such an arrival, the current prefix $T'$ and the preceding prefix have the same player sets in $G$ and $G'$, so the increment is the same. Suppose it is positive. If $i$ is selected in $G'$, the comparison above shows that she is selected in $G$ from the same sharing set with a rank that is no larger, so by (a) her fraction in $G$ is at least her fraction in $G'$. Otherwise she receives nothing in $G'$.

\emph{Arrivals of the pair.} In $G'$, player $i$ is paid only at her own arrival, which creates the increment $\Delta'=v(T)-v(A\cup\{j\})$. If $\Delta'=0$, there is nothing to prove. Otherwise $i$ is the entrant at $T$ in $G'$; let $S'$ be the sharing set and $s'=|S'|$, so that $i$ receives $\mu\Delta'$ with $\mu=\lambda^{(t,s')}_{s'}$. Since $\Delta'>0$, player $i$ is critical in $T$, so $i\in S'$. In $G$, if $v(T)>v(A\cup\{i\})$, the arrival of $j$ creates a positive increment, and the comparison above with $T'=T$ shows that its sharing set is $S'$, with $i$ at rank at most $s'$. By (a), $i$ then receives at least $\mu (v(T)-v(A\cup\{i\}))$ at the arrival of $j$. If $v(A\cup\{i\})\leq v(A\cup\{j\})$, this payment is at least $\mu\Delta'$. If $v(A\cup\{i\})>v(A\cup\{j\})$, then $i$ is incremental at her own arrival in $G$, where $t-1$ players are present. Let $s_1$ be the size of the sharing set there. If $s'\geq 2$, condition (b) gives $\lambda^{(t-1,s_1)}_{s_1}\ge\mu$. If $s'=1$, then $S'=\{i\}$, so $v(\{i\})=v(T)$ and $v(A\cup\{i\})=v(T)$. The deletion procedure on $A\cup\{i\}$ then keeps $i$ and removes every player of $A$, so $s_1=1$ and her fraction is $1=\mu$. In both cases, her payments at the two arrivals in $G$ total at least
\begin{align*}
	\mu\left(v(A\cup\{i\})-v(A)\right)+\mu\left(v(T)-v(A\cup\{i\})\right)=\mu\left(v(T)-v(A)\right)\ge\mu\Delta'.
\end{align*}
Combining the two comparisons gives $\phi(G,i)\ge\phi(G',i)$, which proves \EA.

Conversely, assume that the rule satisfies the four axioms. We use only \PART{} and \EA. For condition (a), fix $1\leq s\leq t$ and consider $t$ players with an arrival order $\pi$. Let $U$ be the set of the last $s$ arrivals, and let $v(S)=1$ if $U\subseteq S$ and $v(S)=0$ otherwise. Only the last arrival creates a positive increment. The deletion procedure keeps every member of $U$, since each of them is critical, and removes every other player. The sharing set is therefore $U$, and the player in position $t-s+r$ has rank $r$. By \PART, the last player receives $\lambda^{(t,s)}_s>0$. For $1\leq r<s$, swap the players in positions $t-s+r$ and $t-s+r+1$. The set $U$ is still the set of the last $s$ arrivals, so the sharing set is unchanged, and the delayed player moves from rank $r$ to rank $r+1$. \EA gives $\lambda^{(t,s)}_r\ge\lambda^{(t,s)}_{r+1}$.

For condition (b), it suffices to treat $t'=t-1$: for $t'<t-1$, apply this case repeatedly through vectors with two selected players,
\begin{align*}
	\lambda^{(t,s)}_s\le\lambda^{(t-1,2)}_2\le\cdots\le\lambda^{(t'+1,2)}_2\le\lambda^{(t',s')}_{s'}.
\end{align*}
If $s'=1$, then $\lambda^{(t-1,1)}_1=1$ and there is nothing to prove. Fix $t\geq 3$, $2\leq s\leq t$, and $2\leq s'\leq t-1$. Consider $t$ players, among them $i$ and $j$. Choose a set $R$ of size $s$ containing $i$ and $j$, and a set $P$ of size $s'$ containing $i$ but not $j$. For $H>0$, let
\begin{align*}
  v_H(S)=
  \begin{cases}
    H+1 & \text{ if } R\subseteq S,\\
    H & \text{ if } P\subseteq S \text{ and } R\not\subseteq S,\\
    0 & \text{otherwise},
  \end{cases}
\end{align*}
which is normalized and monotone. Let $\pi$ order the other players arbitrarily, followed by $j$ and then $i$, and let $\pi'$ swap $i$ and $j$.

Under $\pi$, every coalition without $i$ has value $0$, so $i$ is paid only at her arrival, where the value rises from $0$ to $H+1$. The coalitions with this value are the supersets of $R$, so the inclusion-minimal sharing set is $R$, and $i$ has rank $s$. Her payoff is $\lambda^{(t,s)}_s(H+1)$. Under $\pi'$, player $i$ arrives when $t-1$ players are present and $j$ is absent, so the value rises from $0$ to $H$. Within the current prefix, the coalitions with value $H$ are the supersets of $P$, so the sharing set is $P$, and $i$ has rank $s'$. She receives $\lambda^{(t-1,s')}_{s'}H$, and the arrival of $j$ adds at most $1$. By \EA,
\[
  \lambda^{(t-1,s')}_{s'}H+1\ \ge\ \lambda^{(t,s)}_s(H+1)\ \ge\ \lambda^{(t,s)}_s H .
\]
The weights do not depend on $H$. Dividing by $H$ and letting $H\to\infty$ gives $\lambda^{(t,s)}_s\le\lambda^{(t-1,s')}_{s'}$.
\end{proof}

\begin{proof}[Proof of \Cref{prop:rulms_axioms}]
By \Cref{thm:wulms}, it suffices to show that the weights of \RULMS satisfy conditions (a) and (b). Consider an incremental arrival with $t$ current players and $s$ selected players. \RULMS assigns
\begin{align*}
	\lambda^{(t,s)}_1=\frac{t-s+1}{t},\qquad \lambda^{(t,s)}_r=\frac{1}{t} \quad(2\leq r\leq s).
\end{align*}
When $s=1$, this gives $\lambda^{(t,1)}_1=1$, so the entrant receives the whole increment. For condition (a), all weights are positive. Since $s\leq t$, we have $(t-s+1)/t \geq 1/t$, so the earliest selected player receives at least as much as every other selected player, and all other selected players receive the same fraction $1/t$. Hence $\lambda^{(t,s)}_1\geq \lambda^{(t,s)}_2\geq \cdots\geq \lambda^{(t,s)}_s>0$. For condition (b), let $t>t'$ and $s\geq 2$. Then $\lambda^{(t,s)}_s=1/t$. At $t'$ current players, the entrant's fraction is $\lambda^{(t',s')}_{s'}=1/t'$ if $s'\geq 2$ and $\lambda^{(t',1)}_1=1$ if $s'=1$. In both cases it is at least $1/t'>1/t$, so $\lambda^{(t,s)}_s\le\lambda^{(t',s')}_{s'}$.

Thus \RULMS satisfies \PART, \SSTAY, \EA, and \OD. For the query bound, at each arrival \RULMS queries the value of the current prefix and runs the deletion procedure of \ULMES, which makes at most one value query per current player. The weights depend only on $t$, $s$, and the ranks, so they require no further queries. Hence the total is at most $\sum_{t=1}^{n}(t+1)=O(n^2)$ value queries.
\end{proof}

\paragraph{Lorenz-optimal Upward Lexicographic Marginal Sharing (\LULMS) Rule.} 
\Cref{thm:wulms} leaves considerable freedom in the choice of weights. We call a weight family \emph{admissible} if it satisfies conditions (a) and (b) of \Cref{thm:wulms}, that is, if the \WULMS rule it induces satisfies \PART, \SSTAY, \EA, and \OD. Admissible weights can divide an increment very unevenly. Under \RULMS, every selected player except the earliest one, including the incremental player, receives only $\frac{1}{t}$, while the earliest selected player receives the remaining $\frac{t-s+1}{t}$. This raises the question of which admissible weight family divides each positive increment among the players in $S^i$ as evenly as possible. We measure evenness by the Lorenz order. For two vectors $x=(x_1,\ldots,x_s)$ and $y=(y_1,\ldots,y_s)$ with $\sum_{r=1}^{s}x_r=\sum_{r=1}^{s}y_r=1$, we say that $x$ \emph{Lorenz-dominates} $y$ if, for every $1\leq \ell\leq s$, the sum of the $\ell$ smallest entries of $x$ is at least the sum of the $\ell$ smallest entries of $y$. When both vectors are nonincreasing, this holds if and only if $\sum_{r=1}^{\ell}x_r\le\sum_{r=1}^{\ell}y_r$ for every $1\leq \ell\leq s$.

\begin{definition}[Lorenz optimality]\label{def:lorenz-optimal}
An admissible weight family $\lambda$ is \emph{Lorenz-optimal} if, for all $1\leq s\leq t$ and every admissible weight family $\mu$, the vector $\lambda^{(t,s)}$ Lorenz-dominates $\mu^{(t,s)}$. A \WULMS rule is \emph{Lorenz-optimal} if its weight family is.
\end{definition}

Lorenz dominance is only a partial order, and condition (b) couples the weight vectors for different numbers of current players, so it is not clear a priori that a Lorenz-optimal weight family exists. We show that exactly one does. The key step is an upper bound on the fraction of the incremental player.
 
\begin{lemma}\label{prop:entrant_bound}
Let $\lambda$ be an admissible weight family. Then for all $2\leq s\leq t$, we have $\lambda^{(t,s)}_s\leq c_{t,s}$ where $c_{t,s}=\frac{1}{\max\{s,\,t-1\}}$. 
\end{lemma}
 
\begin{proof}
By condition (a), $\lambda^{(t,s)}_s$ is the smallest entry of $\lambda^{(t,s)}$, so $s\,\lambda^{(t,s)}_s\leq \sum_{r=1}^{s}\lambda^{(t,s)}_r=1$. By condition (b) with $t'=t-1$ and $s'=t-1$, and by the same argument applied to $\lambda^{(t-1,t-1)}$, we have $\lambda^{(t,s)}_s\le\lambda^{(t-1,t-1)}_{t-1}\le\frac{1}{t-1}$.
\end{proof}
 
Intuitively, the bound $\frac{1}{s}$ comes from dividing one unit among $s$ selected players while weakly favoring earlier ones. The bound $\frac{1}{t-1}$ comes from condition (b): with $t$ current players, the incremental player cannot receive more than in a prefix subgame with $t-1$ current players in which all of them are selected, where her fraction is at most $\frac{1}{t-1}$. This suggests giving the incremental player exactly $c_{t,s}$ and dividing the remaining fraction equally among the other selected players.
 
\begin{definition}[\LULMS rule]\label{def:lulms}
The \emph{Lorenz-optimal Upward Lexicographic Marginal Sharing} (\LULMS) rule is the \WULMS rule induced by the weight family with $\lambda^{(t,1)}_1=1$ and, for all $2\leq s\leq t$, $\lambda^{(t,s)}_r=\frac{1-c_{t,s}}{s-1}$ for $1\leq r \leq s-1$ and $\lambda^{(t,s)}_s=c_{t,s}$.
\end{definition}
 
When $s\in\{t-1,t\}$, we have $c_{t,s}=\frac{1}{s}$, so \LULMS divides the increment equally among the players in $S^i$, as \ULMES does. When $2\leq s\leq t-2$, the incremental player receives $\frac{1}{t-1}$, and every other selected player receives the same larger fraction $\frac{t-2}{(t-1)(s-1)}$. For example, with $t=5$ and $s=3$, \RULMS{} uses $(\frac{3}{5},\frac{1}{5} ,\frac{1}{5})$, whereas \LULMS uses $(\frac{3}{8},\frac{3}{8},\frac{1}{4})$. In \Cref{ex:ulmes-not-ea}, player $3$ receives $\frac{x}{3}+\frac{2(y-x)}{3}$ under $\pi=(1,2,3,4)$ and only $\frac{y}{3}$ under $\pi'=(1,2,4,3)$, so there is still no benefit from delaying her arrival. Compared with \RULMS, the incremental player keeps a larger fraction, and the remainder is spread evenly instead of being concentrated on the earliest selected player.

\begin{proposition}\label{prop:lulms}
\LULMS is the unique Lorenz-optimal \WULMS rule. In particular, it satisfies \PART, \SSTAY, \EA, and \OD and can be computed with $O(n^2)$ value queries.
\end{proposition}

\begin{proof}
Write $\lambda$ for the weight family of \LULMS. We first show that $\lambda$ is admissible. For $s\geq 2$, we have $c_{t,s}\le\frac{1}{s}$, so $\frac{1-c_{t,s}}{s-1}\geq c_{t,s}>0$, which gives condition (a). For condition (b), let $t>t'\geq 1$, $2\leq s\leq t$, and $1\leq s'\leq t'$. Then $c_{t,s}\le\frac{1}{t-1}\le\frac{1}{t'}$. With $t'$ current players, the fraction of the incremental player is $\lambda^{(t',1)}_1=1$ if $s'=1$, and $c_{t',s'}=\frac{1}{\max\{s',t'-1\}}\ge\frac{1}{t'}$ if $s'\geq 2$. Hence $\lambda^{(t,s)}_s\le\lambda^{(t',s')}_{s'}$. By \Cref{thm:wulms}, \LULMS therefore satisfies \PART, \SSTAY, \EA, and \OD, and like every \WULMS rule it can be computed with $O(n^2)$ value queries.

We next show that $\lambda$ is Lorenz-optimal. For $s=1$, every weight vector equals $(1)$. Fix $2\leq s\leq t$, let $\mu$ be an admissible weight family, and write $x=\lambda^{(t,s)}$ and $y=\mu^{(t,s)}$. By condition (a), $y$ is nonincreasing, and by \Cref{prop:entrant_bound}, $y_s\leq c_{t,s}$. For $1\le\ell\leq s-1$, the average of the first $\ell$ entries of $y$ is at least the average of its first $s-1$ entries, so
\begin{align*}
  \sum_{r=1}^{\ell}y_r\ \geq \ \frac{\ell}{s-1}\sum_{r=1}^{s-1}y_r\ =\ \frac{\ell\,(1-y_s)}{s-1}\ \geq \ \frac{\ell\,(1-c_{t,s})}{s-1}\ =\ \sum_{r=1}^{\ell}x_r .
\end{align*}
For $\ell=s$, both sums equal one. Since $x$ and $y$ are nonincreasing, $x$ Lorenz-dominates $y$.

Finally, let $\mu$ be any Lorenz-optimal weight family. For all $1\leq s\leq t$, the vectors $\lambda^{(t,s)}$ and $\mu^{(t,s)}$ then Lorenz-dominate each other. Both are nonincreasing by condition (a), so their partial sums coincide for every $\ell$, which gives $\mu^{(t,s)}=\lambda^{(t,s)}$. Hence $\mu=\lambda$.
\end{proof}

Lorenz optimality is a property of the weight vectors, that is, of how each single increment is divided. It does not carry over to cumulative payoffs, because different increments are shared among different sets of players. To illustrate the value-sharing rules studied in this paper, we close by tracing them on one game. We recall \Cref{ex:ulmes-not-ea}, take $x=1$ and $y=2$, and add a dummy player $d$ who arrives first. Formally, $N=\{d,1,2,3,4\}$ and
\begin{align*}
	v(S)=
	\begin{cases}
		2 & \text{if } \{3,4\}\subseteq S, \\
		1 & \text{if } \{1,2,3\}\subseteq S \text{ and } 4\notin S, \\
		0 & \text{otherwise.}
	\end{cases}
\end{align*}
We consider two arrival orders that differ only in the last two players, $(d,1,2,3,4)$ and $(d,1,2,4,3)$. Under both orders, the arrivals of $d$, $1$, and $2$ create no value, so no rule pays anything. \Cref{fig:example} shows, for each rule and each order, what the players hold after each of the last two arrivals: the light part of a bar is what a player held before the arrival, and the dark part is what she has just received.

\paragraph{Order $(d,1,2,3,4)$.} The arrival of $3$ creates the first unit, through the coalition $\{1,2,3\}$. \DMC gives it to the entrant~$3$, and \eRFC gives it to player~$1$, the earliest critical player. \MES divides it equally among the four players present, including the dummy player $d$. The other equal-sharing rules divide it equally among players $1$, $2$, and $3$, who are at once the non-dummy players, the critical players, and the sharing set of \ULMES. \LULMS divides it in the same way, whereas \RULMS gives one half to the earliest player~$1$ and one quarter each to players $2$ and $3$. The arrival of $4$ creates the second unit, through the coalition $\{3,4\}$. \DMC gives it to the entrant~$4$, and \eRFC gives it to player~$3$, the earliest critical player of the new GM component. \MES divides it among all five players, and \NDMES among the four non-dummy players, so players $1$ and $2$ keep receiving payments. \CMES and \ULMES, and under this order also their GM extensions \CES and \DULMES, divide it equally between players $3$ and $4$. \RULMS and \LULMS also pay only players $3$ and $4$, but they favor the earlier one: the entrant~$4$ receives $\frac{1}{5} $ under \RULMS and $\frac{1}{4} $ under \LULMS. The final payoffs of $(d,1,2,3,4)$ are therefore $(0,\frac{1}{3},\frac{1}{3},\frac{13}{12},\frac{1}{4})$ under \LULMS and $(0,\frac{1}{2},\frac{1}{4},\frac{21}{20},\frac{1}{5})$ under \RULMS. The four smallest payoffs sum to $\frac{11}{12}$ under \LULMS but to $\frac{19}{20}$ under \RULMS, so the cumulative payoff vector of \LULMS does not Lorenz-dominate that of \RULMS, although every weight vector of \LULMS Lorenz-dominates the corresponding weight vector of \RULMS.

\paragraph{Order $(d,1,2,4,3)$.} The arrival of $4$ creates no value, because the coalition $\{d,1,2,4\}$ has value~$0$. The arrival of $3$ then creates both units at once. \DMC gives both units to the entrant~$3$. \eRFC gives one unit each to players $3$ and $4$, the earliest critical players of the two GM components. \MES divides the two units equally among all five players, and \NDMES among the four non-dummy players. \CMES and \ULMES divide them equally between players $3$ and $4$, the critical players of the grand coalition. Their GM extensions now differ. \CES gives the first unit to player~$3$ alone, the only critical player of its GM component, and divides the second unit between players $3$ and $4$. \DULMES divides the first unit among players $1$, $2$, and $3$, as under the first order, and the second unit between players $3$ and $4$. \RULMS and \LULMS divide both units between players $4$ and $3$ and again favor the earlier one, who is now player~$4$: the entrant~$3$ receives the fraction $\frac{1}{5} $ under \RULMS and $\frac{1}{4} $ under \LULMS. Comparing the two orders, player~$3$ ends with more under the second order exactly under \DMC, \CMES, \ULMES, and \CES.

\begin{figure}[!htbp]
\centering
\definecolor{pl0l}{HTML}{E6E6E6}\definecolor{pl0d}{HTML}{ABABAB}
\definecolor{pl1l}{HTML}{DADAFA}\definecolor{pl1d}{HTML}{9797EC}
\definecolor{pl2l}{HTML}{D3EFD3}\definecolor{pl2d}{HTML}{7CCB7C}
\definecolor{pl3l}{HTML}{FBE7B3}\definecolor{pl3d}{HTML}{EDB52E}
\definecolor{pl4l}{HTML}{F9D2D2}\definecolor{pl4d}{HTML}{E98686}
\definecolor{entC}{HTML}{2A9D8F}
\def\sw{0.56}
\def\bw{0.18}
\def\uh{0.66}
\def\xA{0}\def\xB{3.2}\def\xC{7.2}\def\xD{10.4}
\newcommand{\frameaxis}[4]{%
	\draw[inkC, line width=0.3pt] (#1,#2) -- (#1+5*\sw,#2);
	\foreach \s/\lab in {#4}{%
		\ifnum\s=#3
			\node[font=\footnotesize, text=entC, inner sep=0pt, anchor=north] at (#1+\s*\sw+0.5*\sw,#2-0.08) {$\boldsymbol{\lab}$};
		\else
			\node[font=\footnotesize, text=inkC, inner sep=0pt, anchor=north] at (#1+\s*\sw+0.5*\sw,#2-0.08) {$\lab$};
		\fi
	}%
}
\newcommand{\onebar}[6]{%
	\pgfmathsetmacro{\xc}{#1+(#3+0.5)*\sw}%
	\pgfmathsetmacro{\ya}{#2+\uh*#4}%
	\pgfmathsetmacro{\yn}{\ya+\uh*#5}%
	\pgfmathparse{#4>0?1:0}\ifnum\pgfmathresult=1 \filldraw[fill=pl#3l, draw=inkC, line width=0.3pt] (\xc-\bw,#2) rectangle (\xc+\bw,\ya);\fi
	\pgfmathparse{#5>0?1:0}\ifnum\pgfmathresult=1
		\filldraw[fill=pl#3d, draw=inkC, line width=0.3pt] (\xc-\bw,\ya) rectangle (\xc+\bw,\yn);
		\node[font=\scriptsize, text=inkC, inner sep=0pt, anchor=south] at (\xc,\yn+0.05) {$+#6$};
	\fi
}
\newcommand{\rulerow}[4]{%
	\node[anchor=east, font=\small, text=inkC, inner sep=0pt, align=right] at (-0.28,#1+0.3) {#2};
	\frameaxis{\xA}{#1}{3}{0/d,1/1,2/2,3/3}
	\frameaxis{\xB}{#1}{4}{0/d,1/1,2/2,3/3,4/4}
	\frameaxis{\xC}{#1}{4}{0/d,1/1,2/2,4/4}
	\frameaxis{\xD}{#1}{3}{0/d,1/1,2/2,3/3,4/4}
	\foreach \s/\ra/\rb/\la/\lb in {#3}{%
		\onebar{\xA}{#1}{\s}{0}{\ra}{\la}
		\onebar{\xB}{#1}{\s}{\ra}{\rb}{\lb}
	}%
	\foreach \s/\rc/\lc in {#4}{%
		\onebar{\xD}{#1}{\s}{0}{\rc}{\lc}
	}%
}
\begin{tikzpicture}[x=1cm, y=1cm]
	\def\yA{0}\def\yB{-1.58}\def\yC{-2.8}\def\yD{-4.1}\def\yE{-5.68}\def\yF{-7.59}\def\yG{-9.06}\def\yH{-11.04}\def\yI{-12.95}
	\node[font=\footnotesize, text=inkC, anchor=south, inner sep=0pt] at (0.5*\xA+0.5*\xB+2.5*\sw,2.6) {Arrival order $(d,1,2,3,4)$};
	\node[font=\footnotesize, text=inkC, anchor=south, inner sep=0pt] at (0.5*\xC+0.5*\xD+2.5*\sw,2.6) {Arrival order $(d,1,2,4,3)$};
	\draw[inkC, line width=0.4pt] (\xA,2.51) -- (\xB+5*\sw,2.51);
	\draw[inkC, line width=0.4pt] (\xC,2.51) -- (\xD+5*\sw,2.51);
	\foreach \x/\who in {\xA/3, \xB/4, \xC/4, \xD/3}{
		\node[font=\footnotesize, text=inkC, anchor=south, inner sep=0pt] at (\x+2.5*\sw,1.95) {{\color{entC}$\boldsymbol{\who}$} joins};
	}
	\draw[univC!70, line width=0.4pt] (0.5*\xB+0.5*\xC+2.5*\sw,3.0) -- (0.5*\xB+0.5*\xC+2.5*\sw,\yI-0.42);
	\rulerow{\yA}{\DMC}{3/1/0/1/, 4/0/1//1}{3/2/2}
	\rulerow{\yB}{\eRFC}{1/1/0/1/, 3/0/1//1}{3/1/1, 4/1/1}
	\rulerow{\yC}{\MES}{0/0.25/0.2/{\frac{1}{4} }/{\frac{1}{5} }, 1/0.25/0.2/{\frac{1}{4} }/{\frac{1}{5} }, 2/0.25/0.2/{\frac{1}{4} }/{\frac{1}{5} }, 3/0.25/0.2/{\frac{1}{4} }/{\frac{1}{5} }, 4/0/0.2//{\frac{1}{5} }}{0/0.4/{\frac25}, 1/0.4/{\frac25}, 2/0.4/{\frac25}, 3/0.4/{\frac25}, 4/0.4/{\frac25}}
	\rulerow{\yD}{\NDMES}{1/0.33333/0.25/{\frac{1}{3} }/{\frac{1}{4} }, 2/0.33333/0.25/{\frac{1}{3} }/{\frac{1}{4} }, 3/0.33333/0.25/{\frac{1}{3} }/{\frac{1}{4} }, 4/0/0.25//{\frac{1}{4} }}{1/0.5/{\frac{1}{2}}, 2/0.5/{\frac{1}{2}}, 3/0.5/{\frac{1}{2}}, 4/0.5/{\frac{1}{2}}}
	\rulerow{\yE}{\CMES, \ULMES}{1/0.33333/0/{\frac{1}{3} }/, 2/0.33333/0/{\frac{1}{3} }/, 3/0.33333/0.5/{\frac{1}{3} }/{\frac{1}{2}}, 4/0/0.5//{\frac{1}{2}}}{3/1/1, 4/1/1}
	\rulerow{\yF}{\CES}{1/0.33333/0/{\frac{1}{3} }/, 2/0.33333/0/{\frac{1}{3} }/, 3/0.33333/0.5/{\frac{1}{3} }/{\frac{1}{2}}, 4/0/0.5//{\frac{1}{2}}}{3/1.5/{\frac{3}{2}}, 4/0.5/{\frac{1}{2}}}
	\rulerow{\yG}{\DULMES}{1/0.33333/0/{\frac{1}{3} }/, 2/0.33333/0/{\frac{1}{3} }/, 3/0.33333/0.5/{\frac{1}{3} }/{\frac{1}{2}}, 4/0/0.5//{\frac{1}{2}}}{1/0.33333/{\frac{1}{3} }, 2/0.33333/{\frac{1}{3} }, 3/0.83333/{\frac56}, 4/0.5/{\frac{1}{2}}}
	\rulerow{\yH}{\RULMS}{1/0.5/0/{\frac{1}{2}}/, 2/0.25/0/{\frac{1}{4} }/, 3/0.25/0.8/{\frac{1}{4} }/{\frac45}, 4/0/0.2//{\frac{1}{5} }}{3/0.4/{\frac25}, 4/1.6/{\frac85}}
	\rulerow{\yI}{\LULMS}{1/0.33333/0/{\frac{1}{3} }/, 2/0.33333/0/{\frac{1}{3} }/, 3/0.33333/0.75/{\frac{1}{3} }/{\frac{3}{4}}, 4/0/0.25//{\frac{1}{4} }}{3/0.5/{\frac{1}{2}}, 4/1.5/{\frac{3}{2}}}
\end{tikzpicture}
\caption{The sharing process on the game for the last two arrivals under two orders. Each row shows one rule, and each frame shows the situation after one arrival; the entrant is highlighted. A bar shows what the player named below it holds. Its light part is what she held before this arrival, and its dark part is what she receives at this arrival, with the amount printed above the bar. A player without a bar holds nothing.}
\label{fig:example}
\end{figure}

\subsection{A Characterization of Ordinal Weighted Sharing Rules}\label{sec:ordinal}
\WULMS rules fix the recipients by the \ULMES procedure. We now characterize more broadly which weighted sharing rules satisfy \PART, \SSTAY, \EA, and \OD, under one restriction on the weights. The rules \MES, \NDMES, \ULMES, and every \WULMS rule decide how to share an increment using only the arrival order and comparisons between coalition values, such as whether a contribution is positive or whether a player can be removed without lowering the current value. The numerical gaps between coalition values affect only the size $\MC_i$ of the increment, not how it is divided. We call such rules ordinal.

\begin{definition}[Ordinal rule]\label{def:ordinal}
A weighted sharing rule induced by $\lambda$ is \emph{ordinal} if $\lambda(G)=\lambda(G')$ for any OCGs $G=(N,v,\pi)$ and $G'=(N,v',\pi)$ whose last arrivals are incremental whenever $v(S) \leq v(T)$ if and only if $v'(S)\leq v'(T)$ for all $S,T \subseteq N$. 
\end{definition}

The comparisons include ties, so an incremental arrival remains incremental, and the condition applies to every prefix subgame. The amounts paid still change with the increments. Being ordinal is a restriction on rule design rather than a consequence of the axioms: \DULMES satisfies all four axioms but is not ordinal. Related ordinal conditions appear in procedural values~\cite{Mala13a}, power indices~\cite{DMO22a}, and portioning~\cite{AAC+23a}; our condition applies to the numerical weights at each prefix, which is stronger than preserving the ranking of the recipients.

We characterize the ordinal rules on the full normalized monotone domain. The support condition captures \PART and \OD, the two monotonicity conditions capture \EA, and \SSTAY follows from positive entry rewards and position monotonicity. Write $\operatorname{supp}\lambda(G)=\{j:\lambda_j(G)>0\}$. For a prefix subgame $G^i$ and a predecessor $j$, $(G^i)_{-j}$ denotes the restricted game obtained by deleting $j$ and preserving the order of the remaining players.

\begin{theorem}
	\label{thm:characterization_weighted_rule}
		An ordinal weighted sharing rule $\phi$ satisfies \EA, \PART, \SSTAY, and \OD if and only if the following conditions hold.
		\begin{enumerate}[label=(\alph*)]
			\item (\textbf{Support condition.}) For every incremental arrival $i$,
			\begin{align*}
				\{i\}\subseteq\supp\lambda(G^i) \subseteq\ND(N_{\pi\mid i}).
			\end{align*}
			\item (\textbf{Position monotonicity.}) For any OCGs $G=(N,v,\pi)$ and $G'=(N,v,\pi')$ whose last arrivals are both incremental, each player $i$ satisfies $\lambda_i(G)\geq\lambda_i(G')$ whenever $\pi_{\mid N\setminus\{i\}}=\pi'_{\mid N\setminus\{i\}}$ and $N_{\pi\mid i}\subsetneq N_{\pi'\mid i}$.
			\item (\textbf{Predecessor monotonicity.}) For every prefix subgame $G^i$ in which $j$ immediately precedes $i$, if $v(N_{\pi\mid i}\setminus\{j\}) > v(N_{\pi\mid i}\setminus\{i\})$, then $\lambda_i(G^i)\leq\lambda_i((G^i)_{-j})$. 
		\end{enumerate}
	\end{theorem}

\begin{proof}
	We first show that $\phi$ satisfies \PART and \OD if and only if the support condition holds. Suppose that $\phi$ satisfies \PART and \OD. By definition, every incremental player $i$ gets positive instant reward which implies $\lambda_i(G^i) > 0$ and $i\in \supp(\lambda(G^i))$. Also, every dummy player receives zero payoff in the prefix subgame $G^i$. Hence, $\supp(\lambda(G^i))$ could only have non-dummy players, i.e., $\supp(\lambda(G^i))\subseteq \ND(N_{\pi\mid i})$. Conversely, the lower support inclusion immediately gives \PART. A player who is dummy in a prefix subgame is dummy in every earlier prefix containing her. The upper inclusion therefore gives her zero at every positive increment, and no payment is made at a zero increment. Hence her cumulative payoff is zero, proving \OD. 

	We next prove that an ordinal weighted sharing rule $\phi$ satisfies \EA on all normalized monotone games if and only if it satisfies position monotonicity and predecessor monotonicity.

	Assume that $\phi$ satisfies \EA. To prove position monotonicity, suppose, for a contradiction, that there exist two OCGs $G=(N,v,\pi)$ and $G'=(N,v,\pi')$ whose last arrivals are both incremental, such that $\pi_{\mid N\setminus\{i\}}=\pi'_{\mid N\setminus\{i\}}$ and $N_{\pi\mid i} \subsetneq N_{\pi'\mid i}$, however, $\lambda_i(G)< \lambda_i(G')$. Let $\delta=\lambda_i(G')-\lambda_i(G)>0$ be the increase in $i$'s sharing weight after delaying, and let $\Delta =\phi(G,i)-\phi(G',i)\geq 0$ be her original payoff difference, where the inequality follows from \EA. 
	
	We then construct another monotone game in which this weight increase makes delaying profitable. We choose $H >\frac{\Delta}{\delta}$ and define the valuation function where for each $S\subseteq N$, 
	\begin{align*}
		v_H(S)= 
		\begin{cases}
			v(S) & v(S) \neq v(N) \\
			v(S) + H & v(S) = v(N).
		\end{cases}
	\end{align*}
	Write $G_H=(N,v_H,\pi)$ and $G'_H=(N,v_H,\pi')$. The valuation $v_H$ is normalized and monotone and preserves all weak comparisons between coalition values. Since $\phi$ is ordinal, the sharing weights at every incremental arrival are therefore unchanged.

	Since the last arrival is incremental in both orders, every proper prefix has value strictly below $v(N)$. Thus only the final increment increases by $H$ in each order, while all earlier payments remain unchanged. So it follows that player $i$ receives an additional $H\lambda_i(G)$ in $G_H$ and an additional $H\lambda_i(G')$ in $G'_H$, which gives us 
	\begin{align*}
		\phi(G_H,i)-\phi(G'_H,i)=\Delta +H\cdot \left(\lambda_i(G)-\lambda_i(G')\right)=\Delta -H\delta < 0.
	\end{align*}
	Player $i$ therefore gains by delaying from $\pi$ to $\pi'$ in the game with valuation $v_H$, contradicting \EA. This proves position monotonicity.

	Regarding predecessor monotonicity, suppose for a contradiction that the ordinal rule $\phi$ satisfies \EA but fails predecessor monotonicity. There exists a prefix subgame $G^i$ with the order $\pi=(\ldots, j,i,\ldots)$ such that $v(N_{\pi\mid i}\setminus \{j\}) > v(N_{\pi\mid i}\setminus \{i\})$ however $\lambda_i(G^i) > \lambda_i((G^i)_{-j})$. Let $T=N_{\pi\mid i}$ and obtain $G'=(T,v_{\mid T},\pi')$ from $G^i$ by swapping its last two players, so that $\pi'$ ends in $(i,j)$. Player $i$ arrives earlier in $G'$, and her prefix subgame at entry is exactly $(G^i)_{-j}$. By definition of \EA, the payoff difference $\Delta=\phi(G',i)-\phi(G^i,i)$ is nonnegative. We choose $H > \frac{\Delta}{\lambda_i(G^i)-\lambda_i((G^i)_{-j})}$ and define for each $S\subseteq N_{\pi \mid i}$ 
	\begin{align*}
    v_H(S)=
	\begin{cases}
		v(S) & v(S) < v(N_{\pi\mid i}\setminus\{j\}) \\
		v(S) + H & v(S) \geq v(N_{\pi\mid i}\setminus\{j\}).
	\end{cases}
	\end{align*}
	Let $G^i_H$ and $G'_H$ be obtained from $G^i$ and $G'$, respectively, by replacing their valuations with $v_H$. The new valuation is normalized and monotone and preserves all weak comparisons between coalition values. Since $\phi$ is ordinal, the sharing weights at every incremental arrival remain unchanged.

	Both orders first reach the threshold $v(N_{\pi\mid i}\setminus\{j\})$ when $i$ arrives. In $G'$, her arrival completes the prefix $(G^i)_{-j}$. In $G^i$, the prefix immediately before her arrival has value $v(N_{\pi\mid i}\setminus\{i\})$, which is strictly below the threshold. Thus $i$'s entry increment increases by $H$ in both games, while her payments at all other arrivals remain unchanged. Therefore we have
	\begin{align*}
    \phi(G'_H,i)-\phi(G^i_H,i) =\Delta +H \left(\lambda_i((G^i)_{-j})-\lambda_i(G^i)\right) <0.
	\end{align*}
	Player $i$ therefore gains by delaying in the game with valuation $v_H$, contradicting \EA. This proves predecessor monotonicity.
	
	For the converse, assume that the two monotonicity conditions hold. By the adjacent-delay reduction in \Cref{lem:EA_adjacent}, it suffices to compare $G=(N,v,\pi)$ with the game $G'=(N,v,\pi')$ obtained by swapping two adjacent players $i\prec_\pi j$. Let $A=N_{\pi\mid i}\setminus\{i\}$ be the common prefix before the pair.

	At every arrival after both players have arrived, the two prefix subgames have the same player set and the same entrant, so they create the same increment. Whenever this increment is positive, position monotonicity makes $i$'s sharing weight weakly larger in $G$. Her payments from all such arrivals are therefore weakly larger in $G$. It remains to compare her payments at the arrivals of $i$ and $j$. If $v(A\cup\{i,j\})=v(A\cup\{j\})$, player $i$ receives zero from these two arrivals in $G'$. Nonnegativity then gives the required comparison. Otherwise, her entry increment in $G'$ is positive.

	If $v(A\cup\{i\})\leq v(A\cup\{j\})$, the increment created by $j$ in $G$ is at least the entry increment of $i$ in $G'$. Both increments are positive and occur on the same player set $A\cup\{i,j\}$. Position monotonicity gives $\lambda_i(G^j)\geq\lambda_i(G'^i)$. Thus $i$'s payment at $j$'s arrival in $G$ alone is at least her payment at her own arrival in $G'$. When $v(A\cup\{i\})>v(A\cup\{j\})$, predecessor monotonicity applied to $G'^i$ gives $\lambda_i(G^i)\geq\lambda_i(G'^i)$. If $j$ creates a positive increment in $G$, position monotonicity also gives $\lambda_i(G^j)\geq\lambda_i(G'^i)$. Every positive increment across the pair in $G$ is therefore shared with $i$ at a fraction at least $\lambda_i(G'^i)$. These increments have total value $v(A\cup\{i,j\})-v(A) \geq v(A\cup\{i,j\})-v(A\cup\{j\})$, where the right-hand side is her entry increment in $G'$. Her total payment across the pair is consequently weakly larger in $G$. Combining this comparison with the payments at all later arrivals, we obtain $\phi(G,i)\geq\phi(G',i)$. Every longer delay is a sequence of adjacent delays, so $\phi$ satisfies \EA. 

	We last prove \SSTAY. Suppose that conditions~(a) and (b) hold. At an incremental arrival $i$, an existing player $j\in\CR(N_{\pi\mid i})$ is critical for the current value. We obtain $G'$ from $G^i$ by moving $j$ to the last position while preserving the relative order of all other players. Player $j$ is incremental in $G'$, so condition~(a) gives $\lambda_j(G')>0$. Both compared last arrivals are incremental, and condition~(b) yields $\lambda_j(G^i)\geq\lambda_j(G')>0$. Hence $\phi(G^i,j)-\phi(G^{i-},j)=\lambda_j(G^i)\MC_i>0$. Together with \STAY, this establishes \SSTAY and completes the proof.
\end{proof}

\begin{remark}\label{rem:ordinal-scope}
Ordinal weighted sharing rules decide who shares an increment, and in what proportions, from the arrival order and the comparisons between coalition values alone: their fractions are unchanged under every transformation of the valuation that preserves these comparisons, and the sizes of coalition values matter only through the size of the increment. The class is broad. Unlike \Cref{thm:wulms}, which fixes the recipients by the \ULMES procedure, it allows any sharing set and any division determined by such comparisons, such as sharing among the non-dummy players, among the critical players, or among an inclusion-minimal coalition generating the current value, with weights that may depend on arrival ranks. In particular, it contains \DMC, \MES, \NDMES, \CMES, \ULMES, and every \WULMS rule, including \RULMS and \LULMS. \Cref{thm:characterization_weighted_rule} thus gives a single test for all of these rules, and for those that fail one of the four axioms, it identifies the failing condition: the support condition for \MES, position monotonicity for \DMC, and predecessor monotonicity for \CMES and \ULMES.
\end{remark}

\begin{remark}\label{rem:ordinal-limits}
Both the ordinal restriction and the monotone domain matter for \Cref{thm:characterization_weighted_rule}. The ordinal class excludes \SV, which violates \STAY and is therefore not a weighted sharing rule, and the GM extensions \eRFC, \CES, and \DULMES, whose weights depend on the gaps between coalition values through the coefficients of the decomposition (\Cref{rem:eulmes-not-ordinal}). For such rules, the conditions of the theorem need not be necessary. For instance, \DULMES satisfies the four axioms but violates position monotonicity: in the game of \Cref{rem:eulmes-not-ordinal}, it gives player~$1$ the weight $\frac{1}{2}$ at the arrival of player~$2$ under the order $(1,2)$, but the larger weight $\frac{a+b}{2b}$ at her own arrival under $(2,1)$. Finally, the \WVS mechanisms of \citet{GZZ26a} and the \IR refinements of Appendix~\ref{app:ir} are defined only on SOCGs and on zero-monotone games, respectively, and therefore lie outside the full monotone domain of the theorem.
\end{remark}

\section{Conclusion and Discussion}\label{sec:conclusion}
We studied participation incentives in online cooperative games and asked which rules satisfy \PART, \SSTAY, \EA, and \OD. We showed that equal-sharing rules can satisfy all four axioms, but no equal-sharing rule that does so can be computed with polynomially many value queries. Weighted sharing avoids this barrier. We characterized the \WULMS rules that satisfy the four axioms, all of which need only $O(n^2)$ value queries, identified \LULMS the unique one among them whose weights divide each increment most evenly in the Lorenz order and finally characterized the general ordinal weighted sharing rules that satisfy the four axioms.

We believe that there are still many interesting open questions. For example, under every \WULMS rule satisfying the four axioms, the entrant's fraction of a shared increment is at most $\frac{1}{t-1}$ (\Cref{prop:entrant_bound}), so dummy players, who receive nothing, can still lower the entrant's payment. \NDMES avoids this effect but needs exponentially many value queries. Is there a rule that satisfies the four axioms, is unaffected by dummy players, and needs only polynomially many value queries? Moreover, \Cref{thm:characterization_weighted_rule} assumes ordinal weights, and \DULMES shows that the four axioms do not force this assumption. Which non-ordinal weighted sharing rules satisfy the four axioms? 

\section*{Acknowledgment}
	This work is partially supported by the NSF-CSIRO grant on ``Fair Sequential Collective Decision-Making" (RG230833) and by the Japan Society for the Promotion of Science under Grant Number 25K21277. The authors thank Sreoshi Banerjee and the anonymous reviewers of AAAI'25, GAIW'25 and AAMAS'26 for their insightful, constructive, and valuable feedback.

\section*{AI Acknowledgment}
	We used LLMs, in particular ChatGPT 6 Sol and Claude Opus 5.5, to improve the exposition of this manuscript and to assist in checking proofs. In particular, Claude Opus 5.5 identified a notation typo in the definition of \SSTAY, an imprecise statement of the theorem on the \IR refinement, and a minor definition issue of the \IR refined version of \DULMES. The \LULMS rule was initially identified by interacting with ChatGPT 6 Sol, and the authors substantially developed the rule and rewrote the whole proof. \Cref{fig:complete_picture}, \Cref{table:summary-online-value-sharing-rules}, and \Cref{fig:example} were developed by Claude Opus 5.5. The authors take full responsibility for all results and arguments.

\bibliographystyle{abbrvnat}
\bibliography{reference-2}

@article{FMM+21a,
  author  = {Flammini, M. and Monaco, G. and Moscardelli, L. and Shalom, M. and Zaks, S.},
  title   = {On the Online Coalition Structure Generation Problem},
  journal = {Journal of Artificial Intelligence Research},
  year    = {2021},
  volume  = {72},
  pages   = {1215--1250}
}

@article{KPS14a,
  author  = {Kash, I. A. and Procaccia, A. D. and Shah, N.},
  title   = {No Agent Left Behind: Dynamic Fair Division of Multiple Resources},
  journal = {Journal of Artificial Intelligence Research},
  year    = {2014},
  volume  = {51},
  pages   = {579--603}
}

@article{CHL+16a,
  author  = {Chen, X. and Hu, X. and Liu, T.-Y. and Ma, W. and Qin, T. and Tang, P. and Wang, C. and Zheng, B.},
  title   = {Efficient Mechanism Design for Online Scheduling},
  journal = {Journal of Artificial Intelligence Research},
  year    = {2016},
  volume  = {56},
  pages   = {429--461}
}

@inproceedings{ZSS+25a,
  author    = {Zhang, Y. and Saha, I. and Sun, Z. and Yokoo, M.},
  title     = {Coalitions on the Fly in Cooperative Games},
  year      = {2025},
  pages     = {1326--1333},
  editor    = {Lynce, I. and Murano, N. and Vallati, M. and Villata, S. and Chesani, F. and Milano, M. and Omicini, A. and Dastani, M.},
  booktitle = {Proceedings of the 28th European Conference on Artificial Intelligence (ECAI)},
  address   = {Bologna, Italy},
  publisher = {IOS Press}
}

@article{Schm69a,
  author  = {Schmeidler, D.},
  title   = {The nucleolus of a characteristic function game},
  journal = {SIAM Journal on Applied Mathematics},
  year    = {1969},
  volume  = {17},
  number  = {6},
  pages   = {1163--1170}
}

@article{Banz64a,
  author  = {Banzhaf, III, J. F.},
  title   = {Weighted voting doesn't work: A mathematical analysis},
  journal = {Rutgers Law Review},
  year    = {1965},
  volume  = {19},
  number  = {2},
  pages   = {317--343}
}

@inproceedings{ZZG+24a,
  author    = {Zhang, J. and Zhang, Y. and Ge, Y. and Zhao, D. and Fu, H. and Tang, Z. G. and Lu, P.},
  title     = {Incentives for Early Arrival in Cost Sharing},
  year      = {2025},
  pages     = {2327--2335},
  editor    = {Vorobeychik, Y. and Das, S. and Now{\'e}, A.},
  booktitle = {Proceedings of the 24th International Conference on Autonomous Agents and Multiagent Systems (AAMAS)},
  address   = {Detroit, MI, USA},
  publisher = {International Foundation for Autonomous Agents and Multiagent Systems}
}

@article{BeVa10a,
  author  = {Bergemann, D. and V{\"a}lim{\"a}ki, J.},
  title   = {The dynamic pivot mechanism},
  journal = {Econometrica},
  year    = {2010},
  volume  = {78},
  number  = {2},
  pages   = {771--789}
}

@article{AtSe13a,
  author  = {Athey, S. and Segal, I.},
  title   = {An efficient dynamic mechanism},
  journal = {Econometrica},
  year    = {2013},
  volume  = {81},
  number  = {6},
  pages   = {2463--2485}
}

@article{BuRoR24,
  author  = {Bullinger, M. and Romen, R.},
  title   = {Stability in online coalition formation},
  journal = {Journal of Artificial Intelligence Research},
  year    = {2025},
  volume  = {82},
  pages   = {2423--2452}
}

@article{Laur22,
  author  = {Doval, L.},
  title   = {Dynamically stable matching},
  journal = {Theoretical Economics},
  year    = {2022},
  volume  = {17},
  number  = {2},
  pages   = {687--724}
}

@book{vNM44a,
  author    = {von Neumann, J. and Morgenstern, O.},
  title     = {Theory of Games and Economic Behavior},
  year      = {1944},
  publisher = {Princeton University Press}
}

@book{PeSu07a,
  author    = {Peleg, B. and Sudh{\"o}lter, P.},
  title     = {Introduction to the Theory of Cooperative Games},
  year      = {2007},
  edition   = {2nd},
  publisher = {Springer}
}

@inproceedings{GZZ+24a,
  author    = {Ge, Y. and Zhang, Y. and Zhao, D. and Tang, Z. G. and Fu, H. and Lu, P.},
  title     = {Incentives for Early Arrival in Cooperative Games},
  year      = {2024},
  pages     = {651--659},
  editor    = {Alechina, N. and Dignum, V. and Dastani, M. and Sichman, J. S.},
  booktitle = {Proceedings of the 23rd International Conference on Autonomous Agents and Multiagent Systems (AAMAS)},
  address   = {Auckland, New Zealand},
  publisher = {International Foundation for Autonomous Agents and Multiagent Systems}
}

@incollection{SHAP53a,
  author    = {Shapley, L. S.},
  title     = {A Value for n-Person Games},
  booktitle = {Contributions to the Theory of Games, Volume II},
  editor    = {Kuhn, H. W. and Tucker, A. W.},
  year      = {1953},
  pages     = {307--318},
  address   = {Princeton},
  publisher = {Princeton University Press}
}

@incollection{GILL59a,
  author    = {Gillies, D. B.},
  title     = {Solutions to general non-zero-sum games},
  booktitle = {Contributions to the Theory of Games, Volume IV},
  editor    = {Tucker, A. W. and Luce, R. D.},
  year      = {1959},
  pages     = {47--86},
  publisher = {Princeton University Press}
}

@incollection{SHUB59a,
  author    = {Shubik, M.},
  title     = {Edgeworth market games},
  booktitle = {Contributions to the Theory of Games, Volume IV},
  editor    = {Tucker, A. W. and Luce, R. D.},
  year      = {1959},
  pages     = {267--278},
  publisher = {Princeton University Press}
}

@incollection{AUMA64a,
  author    = {Aumann, R. J. and Maschler, M.},
  title     = {The bargaining set for cooperative games},
  booktitle = {Advances in Game Theory},
  editor    = {Dresher, M. and Shapley, L. S. and Tucker, A. W.},
  year      = {1964},
  pages     = {443--476},
  series    = {Annals of Mathematics Studies},
  volume    = {52},
  publisher = {Princeton University Press}
}

@article{BeVa19a,
  author  = {Bergemann, D. and V{\"a}lim{\"a}ki, J.},
  title   = {Dynamic mechanism design: An introduction},
  journal = {Journal of Economic Literature},
  year    = {2019},
  volume  = {57},
  number  = {2},
  pages   = {235--274}
}

@incollection{ROSO92a,
  author    = {Roth, A. E. and Sotomayor, M.},
  title     = {Two-sided matching},
  booktitle = {Handbook of Game Theory with Economic Applications},
  editor    = {Aumann, R. J. and Hart, S.},
  year      = {1992},
  pages     = {485--541},
  chapter   = {16},
  volume    = {1},
  publisher = {Elsevier}
}

@book{BDT08a,
  author    = {Branzei, R. and Dimitrov, D. and Tijs, S.},
  title     = {Models in Cooperative Game Theory},
  year      = {2008},
  edition   = {2nd},
  publisher = {Springer}
}

@book{CEW22a,
  author    = {Chalkiadakis, G. and Elkind, E. and Wooldridge, M.},
  title     = {Computational Aspects of Cooperative Game Theory},
  year      = {2022},
  publisher = {Springer}
}

@article{ALRE04a,
  author  = {Alcalde, J. and Revilla, P.},
  title   = {Researching with whom? Stability and manipulation},
  journal = {Journal of Mathematical Economics},
  year    = {2004},
  volume  = {40},
  number  = {8},
  pages   = {869--887}
}

@inproceedings{DEK+12a,
  author    = {Darmann, A. and Elkind, E. and Kurz, S. and Lang, J. and Schauer, J. and Woeginger, G.},
  title     = {Group activity selection problem},
  year      = {2012},
  pages     = {156--169},
  editor    = {Goldberg, P. W.},
  booktitle = {Proceedings of the 8th International Workshop on Internet and Network Economics (WINE)},
  address   = {Liverpool, UK},
  publisher = {Springer}
}

@article{BOW08a,
  author  = {Le Breton, M. and Ortu{\~n}o-Ort{\'i}n, I. and Weber, S.},
  title   = {{Gamson}'s law and hedonic games},
  journal = {Social Choice and Welfare},
  year    = {2008},
  volume  = {30},
  number  = {1},
  pages   = {57--67}
}

@inproceedings{AZBR12a,
  author    = {Aziz, H. and Brandl, F.},
  title     = {Existence of stability in hedonic coalition formation games},
  year      = {2012},
  pages     = {763--770},
  editor    = {Conitzer, V. and Winikoff, M. and Padgham, L. and van der Hoek, W.},
  booktitle = {Proceedings of the 11th International Conference on Autonomous Agents and Multiagent Systems (AAMAS)},
  address   = {Valencia, Spain},
  publisher = {International Foundation for Autonomous Agents and Multiagent Systems}
}

@incollection{AZSA16a,
  author    = {Aziz, H. and Savani, R.},
  title     = {Hedonic Games},
  booktitle = {Handbook of Computational Social Choice},
  editor    = {Brandt, F. and Conitzer, V. and Endriss, U. and Lang, J. and Procaccia, A. D.},
  year      = {2016},
  pages     = {356--376},
  publisher = {Cambridge University Press}
}

@inproceedings{PORT04a,
  author    = {Porter, R.},
  title     = {Mechanism design for online real-time scheduling},
  year      = {2004},
  pages     = {61--70},
  editor    = {Breese, J. S. and Feigenbaum, J. and Seltzer, M. I.},
  booktitle = {Proceedings of the 5th ACM Conference on Electronic Commerce (EC)},
  address   = {New York, NY, USA},
  publisher = {ACM}
}

@book{HEBE06a,
  author    = {Van Hentenryck, P. and Bent, R.},
  title     = {Online stochastic combinatorial optimization},
  year      = {2006},
  publisher = {The MIT Press}
}

@book{BOEL05a,
  author    = {Borodin, A. and El-Yaniv, R.},
  title     = {Online Computation and Competitive Analysis},
  year      = {2005},
  publisher = {Cambridge University Press}
}

@inproceedings{LANI00a,
  author    = {Lavi, R. and Nisan, N.},
  title     = {Competitive analysis of incentive compatible on-line auctions},
  year      = {2000},
  pages     = {233--241},
  editor    = {Jhingran, A. and MacKie-Mason, J. and Tygar, D. J.},
  booktitle = {Proceedings of the 2nd ACM Conference on Electronic Commerce (EC)},
  address   = {Minneapolis, MN, USA},
  publisher = {ACM}
}

@inproceedings{FRPA03a,
  author    = {Friedman, E. J. and Parkes, D. C.},
  title     = {Pricing {WiFi} at {Starbucks}: issues in online mechanism design},
  year      = {2003},
  pages     = {240--241},
  editor    = {Menasc{\'e}, D. A. and Nisan, N.},
  booktitle = {Proceedings of the 4th ACM Conference on Electronic Commerce (EC)},
  address   = {San Diego, CA, USA},
  publisher = {ACM}
}

@inproceedings{PASI03a,
  author    = {Parkes, D. C. and Singh, S. P.},
  title     = {An {MDP}-Based Approach to Online Mechanism Design},
  year      = {2003},
  pages     = {791--798},
  editor    = {Thrun, S. and Saul, L. K. and Sch{\"o}lkopf, B.},
  booktitle = {Proceedings of the 17th Annual Conference on Neural Information Processing Systems (NIPS)},
  address   = {Vancouver, BC, Canada},
  publisher = {MIT Press}
}

@inproceedings{PYS04a,
  author    = {Parkes, D. C. and Singh, S. P. and Yanovsky, D.},
  title     = {Approximately efficient online mechanism design},
  year      = {2004},
  pages     = {1049--1056},
  editor    = {Saul, L. K. and Weiss, Y. and Bottou, L.},
  booktitle = {Proceedings of the 18th Annual Conference on Neural Information Processing Systems (NIPS)},
  address   = {Vancouver, BC, Canada},
  publisher = {MIT Press}
}

@inproceedings{KVV90a,
  author    = {Karp, R. M. and Vazirani, U. V. and Vazirani, V. V.},
  title     = {An optimal algorithm for on-line bipartite matching},
  year      = {1990},
  pages     = {352--358},
  editor    = {Ortiz, H.},
  booktitle = {Proceedings of the 22nd Annual ACM Symposium on Theory of Computing (STOC)},
  address   = {Baltimore, MD, USA},
  publisher = {ACM}
}

@article{KAPR00a,
  author  = {Kalyanasundaram, B. and Pruhs, K. R.},
  title   = {An optimal deterministic algorithm for online b-matching},
  journal = {Theoretical Computer Science},
  year    = {2000},
  volume  = {233},
  number  = {1-2},
  pages   = {319--325}
}

@inproceedings{FMM09a,
  author    = {Feldman, J. and Mehta, A. and Mirrokni, V. and Muthukrishnan, S.},
  title     = {Online stochastic matching: Beating 1-1/e},
  year      = {2009},
  pages     = {117--126},
  editor    = {Spielman, D. A.},
  booktitle = {Proceedings of the 50th Annual IEEE Symposium on Foundations of Computer Science (FOCS)},
  address   = {Atlanta, GA, USA},
  publisher = {IEEE}
}

@inproceedings{GKM+19a,
  author    = {Gamlath, B. and Kapralov, M. and Maggiori, A. and Svensson, O. and Wajc, D.},
  title     = {Online matching with general arrivals},
  year      = {2019},
  pages     = {26--37},
  editor    = {Zuckerman, D.},
  booktitle = {Proceedings of the 60th Annual IEEE Symposium on Foundations of Computer Science (FOCS)},
  address   = {Baltimore, MD, USA},
  publisher = {IEEE}
}

@incollection{PARK07a,
  author    = {Parkes, D. C.},
  title     = {Online mechanisms},
  booktitle = {Algorithmic Game Theory},
  editor    = {Nisan, N. and Roughgarden, T. and Tardos, {\'E}. and Vazirani, V. V.},
  year      = {2007},
  pages     = {411--440},
  chapter   = {16},
  publisher = {Cambridge University Press}
}

@article{BURO23a,
  author  = {Bullinger, M. and Romen, R.},
  title   = {Online coalition formation under random arrival or coalition dissolution},
  journal = {ACM Transactions on Algorithms},
  year    = {2026},
  volume  = {22},
  number  = {1},
  pages   = {4:1--4:43}
}

@article{BBWa23a,
  author  = {Brandt, F. and Bullinger, M. and Wilczynski, A.},
  title   = {Reaching individually stable coalition structures},
  journal = {ACM Transactions on Economics and Computation},
  year    = {2023},
  volume  = {11},
  number  = {1-2},
  pages   = {4:1--4:65}
}

@article{GZZ26a,
  author  = {Ge, Y. and Zhang, Y. and Zhao, D.},
  title   = {Fair Incentives for Early Arrival in 0-1 Cooperative Games},
  journal = {Proceedings of the AAAI Conference on Artificial Intelligence},
  year    = {2026},
  volume  = {40},
  number  = {20},
  pages   = {16938--16945}
}

@incollection{Gu26a,
  author    = {Gu, S.},
  title     = {Generalizing {I4EA} Mechanisms in Online Cooperative Games via 0-1 Decomposition},
  booktitle = {Mechanism Design in Social Networks},
  editor    = {Hao, D. and Li, B. and Nath, S. and Todo, T. and Zhao, D.},
  year      = {2026},
  pages     = {22--36},
  series    = {Communications in Computer and Information Science},
  volume    = {2777},
  address   = {Singapore},
  publisher = {Springer}
}

@article{ZGT+26a,
  author  = {Zhao, D. and Ge, Y. and Zhang, Y. and Tang, Z. G. and Fu, H. and Lu, P.},
  title   = {Incentives for early arrival in online cooperative games},
  journal = {Artificial Intelligence},
  year    = {2026},
  volume  = {355},
  pages   = {104523}
}

@inproceedings{AGS26a,
  author    = {Aziz, H. and Guo, Y. and Sun, Z.},
  title     = {Participation Incentives in Online Cooperative Games},
  year      = {2026},
  pages     = {478--486},
  editor    = {Mascardi, V. and Thangarajah, J. and Amato, C. and Dennis, L. A.},
  booktitle = {Proceedings of the 25th International Conference on Autonomous Agents and Multiagent Systems (AAMAS)},
  address   = {Paphos, Cyprus},
  publisher = {International Foundation for Autonomous Agents and Multiagent Systems}
}

@article{AdKr17a,
  author  = {Adam, L. and Kroupa, T.},
  title   = {The intermediate set and limiting superdifferential for coalitional games: between the core and the {Weber} set},
  journal = {International Journal of Game Theory},
  year    = {2017},
  volume  = {46},
  number  = {4},
  pages   = {891--918}
}

@article{Mala13a,
  author  = {Malawski, M.},
  title   = {``Procedural'' values for cooperative games},
  journal = {International Journal of Game Theory},
  year    = {2013},
  volume  = {42},
  number  = {1},
  pages   = {305--324}
}

@inproceedings{DMO22a,
  author    = {Doignon, J.-P. and Moretti, S. and {\"O}zt{\"u}rk, M.},
  title     = {On the Ordinal Invariance of Power Indices on Coalitional Games},
  year      = {2022},
  pages     = {258--264},
  editor    = {De Raedt, L.},
  booktitle = {Proceedings of the 31st International Joint Conference on Artificial Intelligence (IJCAI)},
  address   = {Vienna, Austria},
  publisher = {International Joint Conferences on Artificial Intelligence Organization}
}

@article{AAC+23a,
  author  = {Airiau, S. and Aziz, H. and Caragiannis, I. and Kruger, J. and Lang, J. and Peters, D.},
  title   = {Portioning using ordinal preferences: Fairness and efficiency},
  journal = {Artificial Intelligence},
  year    = {2023},
  volume  = {314},
  pages   = {103809}
}

@misc{WRZ+20a,
  author       = {Wang, T. and Rausch, J. and Zhang, C. and Jia, R. and Song, D.},
  title        = {A Principled Approach to Data Valuation for Federated Learning},
  year         = {2020},
  howpublished = {arXiv preprint arXiv:2009.06192},
  url          = {https://arxiv.org/abs/2009.06192}
}

@misc{Rose11a,
  author       = {Rosenfeld, M.},
  title        = {Analysis of {B}itcoin Pooled Mining Reward Systems},
  year         = {2011},
  howpublished = {arXiv preprint arXiv:1112.4980},
  url          = {https://arxiv.org/abs/1112.4980}
}

@inproceedings{SBBR16a,
  author    = {Schrijvers, O. and Bonneau, J. and Boneh, D. and Roughgarden, T.},
  title     = {Incentive Compatibility of {B}itcoin Mining Pool Reward Functions},
  year      = {2017},
  pages     = {477--498},
  editor    = {Grossklags, J. and Preneel, B.},
  booktitle = {Proceedings of the 20th International Conference on Financial Cryptography and Data Security (FC)},
  address   = {Christ Church, Barbados},
  publisher = {Springer}
}

\appendix

\section{Early Arrival and Local Early Arrival}\label{app:adjacent}
\Cref{lem:EA_adjacent} reduces \EA to delays by one position. The reduction is used in~\citep{GZZ+24a,ZZG+24a,GZZ26a}; we include a proof for completeness.

\EAAdjacent* 

\begin{proof}
An adjacent swap is a special delay: if $\pi=(\dots,i,j,\dots)$ and $\pi'=(\dots,j,i,\dots)$, then $\pi_{\mid N\setminus\{i\}}=\pi'_{\mid N\setminus\{i\}}$ and $N_{\pi\mid i}\subset N_{\pi'\mid i}$. Hence \EA implies the condition in \Cref{lem:EA_adjacent}.

Conversely, assume the condition, and let $\pi_{\mid N\setminus\{i\}}=\pi'_{\mid N\setminus\{i\}}$ and $N_{\pi\mid i}\subset N_{\pi'\mid i}$. The two orders agree on all players other than $i$, and $i$ arrives strictly later under $\pi'$. Let $j_1,\dots,j_k$, with $k\geq 1$, be the players who arrive after $i$ under $\pi$ and before $i$ under $\pi'$, listed in their common order. Denote by $\pi=(B,i,j_1,\dots,j_k,C)$ and $\pi'=(B,j_1,\dots,j_k,i,C)$ the orders of the two OCGs where $B$ and $C$ are some sequences of the remaining players. For $0\leq r\leq k$, let $\pi_r=(B,j_1,\dots,j_r,i,j_{r+1},\dots,j_k,C)$ and $G_r=(N,v,\pi_r)$. So that $\pi_0=\pi$ and $\pi_k=\pi'$. For $1\leq r\leq k$, the order $\pi_r$ is obtained from $\pi_{r-1}$ by swapping $i$ with $j_r$, who arrives immediately after her under $\pi_{r-1}$. The condition gives $\phi(G_{r-1},i)\geq\phi(G_r,i)$, and therefore $\phi(G,i)=\phi(G_0,i)\geq\phi(G_1,i)\geq\cdots\geq\phi(G_k,i)=\phi(G',i)$. This completes the proof.
\end{proof}

\section{Greedy Monotone (GM) Decomposition}\label{app:gm}
We revisit the Greedy Monotone (GM) decomposition algorithm proposed by \citet{GZZ+24a} and developed further by \citet{ZGT+26a}. The GM decomposition maps a general OCG $G=(N,v,\pi)$ into a linear combination of SOCGs. The procedure is given in \Cref{alg:gm}.

\begin{algorithm}[!htbp]
\caption{Greedy Monotone (GM) Decomposition}
\label{alg:gm}
\begin{algorithmic}[1]
\REQUIRE An OCG $G=(N,v,\pi)$.
\ENSURE A decomposition $D(G)$.
\STATE Initialize $D(G)\leftarrow\emptyset$ and $k\leftarrow 1$.
\WHILE{$\max_{T\subseteq N}v(T)>0$}
\STATE $S\leftarrow\arg\min_{T\subseteq N,\,v(T)>0}v(T)$.
\STATE Coefficient $c_k\leftarrow v(S)$.
\STATE Initialize the component SOCG $\bar G_k=(N,\bar v_k,\pi)$.
\FOR{$T\subseteq N$}
\IF{$v(T)>0$}
\STATE $\bar v_k(T)\leftarrow 1$.
\ELSE
\STATE $\bar v_k(T)\leftarrow 0$.
\ENDIF
\STATE Update $v(T)\leftarrow v(T)-c_k\bar v_k(T)$.
\ENDFOR
\STATE Add the component $(c_k,\bar G_k)$ to $D(G)$.
\STATE Update $k\leftarrow k+1$.
\ENDWHILE
\end{algorithmic}
\end{algorithm}

\Cref{alg:gm} first initializes $D(G)=\emptyset$ and an iteration counter $k=1$. While some subset of $N$ has a positive value, it identifies a subset $S$ with the minimum value among all subsets with positive value, and sets the coefficient of the $k$-th component game to $c_k=v(S)$. To determine the valuation function $\bar v_k$ of the $k$-th component game, it assigns $\bar v_k(T)=1$ to every subset $T\subseteq N$ with $v(T)>0$, and $\bar v_k(T)=0$ to every other subset. It then updates the valuation function by setting $v(T)=v(T)-c_k$ for every subset $T$ with positive value. This completes the creation of the $k$-th component game $\bar G_k=(N,\bar v_k,\pi)$ with coefficient $c_k$. The following example shows how the GM decomposition runs.

\begin{example}\label{ex:gm}
Consider an OCG $G=(N,v,\pi)$ with $N=\{1,2\}$, $v(\{1\})=v(\{2\})=1$, $v(\{1,2\})=5$, and $\pi=(1,2)$. We first initialize $D(G)=\emptyset$ and $k=1$. For the first component game $\bar G_1=(N,\bar v_1,\pi)$, the minimum positive coalition value is $1$, so $c_1=1$. Since all three coalitions have positive value, we have $\bar v_1(\{1\})=\bar v_1(\{2\})=\bar v_1(\{1,2\})=1$, and the update gives $v(\{1\})=v(\{2\})=0$ and $v(\{1,2\})=4$. Next, $k=2$. For the second component game $\bar G_2=(N,\bar v_2,\pi)$, the minimum positive coalition value is $v(\{1,2\})=4$, so $c_2=4$. Only $\{1,2\}$ has positive value, so $\bar v_2(\{1\})=\bar v_2(\{2\})=0$ and $\bar v_2(\{1,2\})=1$, and the update gives $v(\{1,2\})=0$. No coalition has positive value any more, so the algorithm terminates and returns $D(G)=\{(1,\bar G_1),(4,\bar G_2)\}$.
\end{example}

\citet{GZZ+24a} prove the following properties of the GM decomposition.
\begin{itemize}
	\item The GM decomposition outputs $D(G)$ such that $v(T)=\sum_{k}c_k\bar v_k(T)$ for each $T\subseteq N$, and every $\bar v_k$ is a 0--1 valued monotone function.
	\item For every player $i$ and prefix subgame $G^i$, the decomposition $D(G^i)$ is consistent with $D(G)$ restricted to the players in $N_{\pi\mid i}$ (consistency between the global game and prefix subgames).
\end{itemize}

We also use the following description of the components, which follows from \Cref{alg:gm}~\citep[Section~5.3]{ZGT+26a}. Let $z_1<\dots<z_K$ be the distinct positive values of $v$, and let $z_0=0$. The GM decomposition has exactly $K$ components, and the $k$-th component has the coefficient $c_k=z_k-z_{k-1}$ and the valuation
\begin{align*}
	\bar{v}_k(T)= \begin{cases}
		1 & \text{if } v(T)\geq z_k,\\
		0 & \text{otherwise.}
	\end{cases}
\end{align*}
We call $z_k$ the \emph{threshold} of the $k$-th component: a coalition is winning in $\bar{G}_k$ exactly when its value in $v$ reaches $z_k$. In \Cref{ex:gm}, the thresholds are $z_1=1$ and $z_2=5$. To see this, note that each iteration of \Cref{alg:gm} subtracts the smallest remaining positive value from every coalition that still has positive value, so after $k-1$ iterations every coalition has lost $z_{k-1}$ or has reached $0$; the coalitions that remain positive are those with $v(T)>z_{k-1}$, that is, with $v(T)\geq z_k$. Three facts follow. The decomposition depends only on $v$ and not on $\pi$. A dummy player of a coalition $T$ in $v$ is a dummy player of $T$ in every $\bar v_k$, because $v(S\cup\{j\})=v(S)$ implies $\bar v_k(S\cup\{j\})=\bar v_k(S)$. Finally, the number of components equals the number of distinct positive values of $v$, which can be exponential in $n$. This last fact bounds the size of $D(G)$ but does not by itself imply a query lower bound for a rule defined through it: \eRFC can be computed with $O(n^2)$ value queries~\citep[Theorem~7]{ZGT+26a}, and so can \CES (\Cref{prop:ces-poly}), whereas \DULMES needs exponentially many (\Cref{prop:ulmes-query-lower}).

\paragraph{Extending a rule from SOCGs to general OCGs.}
\citet{GZZ+24a} use the GM decomposition to generalize the \RFC rule for SOCGs to the \eRFC rule for general OCGs. The same construction applies to every online value-sharing rule $\rho$ defined on SOCGs. Given an OCG $G=(N,v,\pi)$, it proceeds as follows:
\begin{enumerate}
	\item[(1)] apply the GM decomposition to decompose $G$ into the component games in $D(G)$;
	\item[(2)] execute the rule $\rho$ within each component game;
	\item[(3)] aggregate the weighted results across all components to obtain each player's total share.
\end{enumerate}
Formally, for a prefix subgame $G^i$, let $\bar G_k^{\,i}$ denote the prefix subgame of $\bar G_k$ with player set $N_{\pi\mid i}$, in which $\rho$ pays nothing if $\bar v_k(N_{\pi\mid i})=0$. The resulting rule $\phi$, which we call the \emph{GM extension} of $\rho$, is given by $\phi(G^i,j)=\sum_{k=1}^{K}c_k\,\rho(\bar G_k^{\,i},j)$ for every $j\in N_{\pi\mid i}$. 

By the consistency property, $\phi(G^i,\cdot)$ depends only on $G^i$, so $\phi$ is an online rule. \eRFC is the GM extension of \RFC. In the same way, \DULMES is the GM extension of \ULMES (Appendix~\ref{app:eulmes}), and \CES is the GM extension of \CMES (Appendix~\ref{app:cmes}).

\begin{remark}\label{rem:gm-preservation}
We remark that the proofs of \citet[Theorems~5--6]{ZGT+26a} for \eRFC apply to the GM extension $\phi$ of every rule $\rho$. To see this, since $v(T)=\sum_k c_k\bar v_k(T)$, we have $\sum_j\phi(G^i,j)=\sum_k c_k\,\bar v_k(N_{\pi\mid i})=v(N_{\pi\mid i})$, so $\phi$ is efficient in every prefix. If $\rho$ satisfies \STAY on SOCGs, every term $\rho(\bar G_k^{\,i},j)$ is non-decreasing as further players arrive, so $\phi$ satisfies \STAY. If $\rho$ satisfies \EA on SOCGs and $G'=(N,v,\pi')$ is obtained from $G$ by delaying $i$, then the two games have the same components up to the order, and $\phi(G,i)-\phi(G',i)=\sum_k c_k\left(\rho(\bar{G}_k,i)-\rho(\bar{G'}_k,i)\right)\geq 0$, where $\bar{G'}_k=(N,\bar v_k,\pi')$. If $\rho$ satisfies \SF on SOCGs, then averaging over all arrival orders and using the linearity of the Shapley value gives \SF for $\phi$.
\end{remark}

We further show that GM decomposition also preserves our newly proposed participation incentive axioms. 
\begin{lemma}\label{lem:gm-participation}
Let $F$ be one of the axioms \PART, \SSTAY, and \OD. If a rule $\rho$ satisfies $F$ on every SOCG, then its GM decomposition extension, denoted by $\phi$, satisfies $F$ on every OCG.
\end{lemma}
\begin{proof}
	Let $i$ be an incremental player, let $T=N_{\pi\mid i}$, and choose $k^\ast$ with $z_{k^\ast}=v(T)$. In the component $\bar{G}_{k^\ast}$, the coalition $T$ is winning and $T\setminus\{j\}$ is losing for every $j\in\CR(T)$, since $v(T\setminus\{j\})<v(T)$. In particular, $T\setminus\{i\}$ is losing, so $i$ is the pivotal player of this component, and every critical player of $T$ in $v$ is critical in it. If $\rho$ satisfies \PART, then $\phi(G^i,i)\geq c_{k^\ast}\rho(\bar G_{k^\ast}^{\,i},i)>0$. If $\rho$ satisfies \SSTAY, then for every existing player $j\in\CR(T)$, the payment to $j$ at the arrival of $i$ is nonnegative in every component by \STAY and positive in the component $\bar G_{k^\ast}$, so $\phi(G^i,j)>\phi(G^{i-},j)$; \STAY of $\phi$ holds by \Cref{rem:gm-preservation}. Finally, recall the threshold description. A dummy player $j$ of $G^i$ is a dummy player of every component restricted to $N_{\pi\mid i}$, so \OD of $\rho$ gives zero in every term of $\phi(G^i,j)$.
\end{proof}

\section{Computing \NDMES}\label{app:ndmes-computation}
\NDMES shares each increment equally among the current non-dummy players, so computing it requires knowing who is a dummy. Whether a player is a dummy can depend on her contribution to any coalition of the current prefix. A single coalition can therefore turn a dummy into a non-dummy without changing the value of any other coalition, and this changes the number of players who share the next increment. The following lower bound builds a game with exponentially many such coalitions, each of which an exact algorithm must query.

\begin{proposition}\label{prop:ndmes-query-lower}
For $n \geq 4$, computing \NDMES requires at least $\frac{2^{n-2}}{n-1}$ value queries.
\end{proposition}

\begin{proof}
Consider an OCG $G=(N,v,\pi)$ where $N=\{1,\ldots,n\}$ and $\pi=(1,\ldots,n)$, so that player~$n$ arrives last, let $r=\lceil \frac{n-2}{2} \rceil$, and for every $S\subseteq N$, let
\begin{align*}
v(S)=
\begin{cases}
  2 & \text{if } n\in S \text{ and } |S\setminus\{n-1,n\}| >r,\\
  1 & \text{if } n\in S \text{ and } 1\leq |S\setminus\{n-1,n\}|\leq r,\\
  1 & \text{if } n\notin S \text{ and } |S\setminus\{n-1,n\}| >r,\\
  0 & \text{otherwise.}
\end{cases}
\end{align*}
In words, a coalition receives one unit of value if it contains player~$n$ and at least one of the players $1,\ldots,n-2$, and another unit if it contains more than $r$ of the players $1,\ldots,n-2$. Player~$n-1$ appears in neither condition, so she is a dummy. Both conditions are preserved when a coalition grows, and two disjoint coalitions cannot both satisfy either of them: only one can contain player~$n$, and only one can contain more than $r$ of the players $1,\ldots,n-2$, since $2(r+1)>n-2$. Hence $v$ is monotone and superadditive, and it has zero singleton values.

For every set $T$ of exactly $r$ players among $1,\ldots,n-2$, let $v_T$ agree with $v$ except that $v_T(T\cup\{n-1\})=1$ instead of $0$. Under $v_T$, player~$n-1$ is no longer a dummy, since adding her to $T$ raises the value from $0$ to $1$. The game $v_T$ is still monotone with zero singleton values, because every proper subset of $T\cup\{n-1\}$ has value $0$ and every proper superset has value at least $1$. It is also still superadditive. Since only the value of $T\cup\{n-1\}$ increases, it suffices to consider a coalition disjoint from $T\cup\{n-1\}$. Such a coalition contains at most $n-2-r\leq r$ of the players $1,\ldots,n-2$, so its value is at most $1$. If its value is $1$, it contains player~$n$ and one of the players $1,\ldots,n-2$ outside $T$, so its union with $T\cup\{n-1\}$ has value $2$; if its value is $0$, the union has value at least $1$ by monotonicity.

All other players are non-dummies in every one of these games, as witnessed by coalitions that do not contain player~$n-1$. A player $i\leq n-2$ raises the value of a set of $r$ other players among $1,\ldots,n-2$ from $0$ to $1$; such a set exists since $r\leq n-3$. Player~$n$ raises the value of $\{i\}$ from $0$ to $1$ for any $i\leq n-2$. When player~$n$ arrives, the value rises from $1$ to $2$ in every game, since $n-2>r$ and the modified coalition $T\cup\{n-1\}$ is neither $N$ nor $N\setminus\{n\}$. Under $v$, this increment is shared among $n-1$ non-dummy players, and player~$n$ receives $\frac{1}{n-1}$. Under every $v_T$, player~$n-1$ is a non-dummy as well, and player~$n$ receives $\frac{1}{n}$. If a deterministic algorithm that computes \NDMES exactly never queries $T\cup\{n-1\}$, it receives the same answers on $v$ and $v_T$ and returns the same output, which is wrong for one of them. Therefore, on input $v$, it must query $T\cup\{n-1\}$ for every $T$, that is, $\binom{n-2}{r}$ coalitions. Since $r$ indexes a largest binomial coefficient in row $n-2$, we have $\binom{n-2}{r} \geq \frac{1}{n-1}\sum_{k=0}^{n-2}\binom{n-2}{k} = \frac{2^{n-2}}{n-1}$. This concludes the query lower bound proof.
\end{proof}

On two structured domains, however, whether a player is a dummy can be read off from few values of $v$, and \NDMES becomes easy to compute.

\begin{proposition}\label{poly_time_NDMES_subadditive}
\NDMES can be computed with $O(n)$ value queries on subadditive valuations and with $O(n^2)$ value queries on supermodular valuations.
\end{proposition}
\begin{proof}
If $v$ is subadditive, the marginal contribution of a player $j$ to any coalition is at most $v(\{j\})$. Hence a player with $v(\{j\})=0$ is a dummy in every prefix that contains her, and a player with $v(\{j\})>0$ is a dummy in none, since her marginal contribution to the empty coalition is positive. The $n$ singleton values therefore determine the non-dummy players of every prefix, and together with the $n$ prefix values they determine the allocation. This uses $O(n)$ value queries.

If $v$ is supermodular, the marginal contribution of a player does not decrease as the coalition grows. A non-dummy of a prefix has a positive marginal contribution to some subset of the prefix without her, and hence also to the prefix without her, so she is critical. Since every critical player is a non-dummy, the non-dummy players of each prefix are exactly its critical players, and \NDMES coincides with \CMES. At each arrival, it therefore suffices to query the value of the current prefix and of its single-player deletions, which uses $O(n)$ value queries per arrival and $O(n^2)$ in total.
\end{proof}

\section{\CMES and Its GM Extension \CES}\label{app:cmes}

Recall \CMES which distributes each increment equally among the critical players  $\CR(N_{\pi \mid i})$ of the current prefix subgame $G^i$. By \Cref{prop:MES_NDMES_axiom}, it satisfies \PART, \SSTAY, and \OD and needs $O(n^2)$ value queries, but fails \EA. It also fails \SF. To see this, let $N=\{1,2\}$, $v(\{1\})=1$, $v(\{2\})=0$, and $v(N)=2$. Under $(1,2)$, player~$1$ receives $1$ at her arrival and half of the increment $1$ at the arrival of player~$2$, since both players are critical in $N$. Under $(2,1)$, the arrival of player~$2$ creates nothing, and the increment $2$ at the arrival of player~$1$ is shared equally between the two critical players. Player~$1$ therefore receives $\frac{5}{4}$ on average, whereas her Shapley value is $\frac{3}{2}$. On this game, \ULMES makes the same payments as \CMES, because at every incremental arrival all current players are critical; hence \ULMES also fails \SF.

On SOCGs, \CMES is the equal sharing among critical players that \citet[Section~6.1, Table~6]{ZGT+26a} illustrate and \citet{GZZ26a} study; the latter show that it fails \EA even on solvable SOCGs~\citep[Example~3]{GZZ26a}. Moreover, its GM extension, which we call the \emph{Critical Equal Share} (\CES) rule, restores \SF while keeping the participation axioms and a polynomial number of value queries. The two rules differ on general games: in the game above under $(2,1)$, \CMES pays $1$ to each player, whereas \CES pays $\frac{3}{2}$ to player~$1$ and $\frac{1}{2}$ to player~$2$. Since \STAY, \EA, and \SF cannot hold together in general~\citep{GZZ+24a}, the failure of \EA is unavoidable for a rule satisfying \SF and \STAY; among such rules, \CES provides \PART and \SSTAY, which \eRFC lacks.

\begin{proposition}\label{prop:ces}
\CES satisfies \SF, \STAY, \SSTAY, \PART, and \OD, but fails \EA.
\end{proposition}

\begin{proof}
By \Cref{rem:gm-preservation,lem:gm-participation}, it suffices to prove \SF, \STAY, \SSTAY, \PART, and \OD for \CMES on SOCGs. Let $G$ be an SOCG with pivotal prefix $U$. \CMES pays only at the arrival of the pivotal player, where each member of $\CR(U)$ receives $\frac{1}{|\CR(U)|}$; this gives \STAY. The pivotal player is the only incremental player, and she and every critical arrived player belong to $\CR(U)$, which gives \PART and \SSTAY. A player who is a dummy in a prefix subgame containing $U$ is a dummy in $U$ and hence not critical in $U$, and nothing is paid before $U$ is complete; this gives \OD.

For \SF, let $\pi$ be drawn uniformly at random, and let $U_\pi$ and $p_\pi$ be its pivotal prefix and pivotal player. For every coalition $U$ that occurs as $U_\pi$ with positive probability, the orders with $U_\pi=U$ are those in which the players of $U$ arrive first, in an order whose last player $\ell$ satisfies that $U\setminus\{\ell\}$ is losing, that is, $\ell\in\CR(U)$, followed by the remaining players in any order. Each $\ell\in\CR(U)$ is last in the same number of these orders, so, conditional on $U_\pi=U$, the pivotal player is uniformly distributed over $\CR(U)$. This is exactly the share that \CMES gives each player when $U_\pi=U$. Taking expectations, the expected share of each player $j$ equals the probability that $p_\pi=j$, which is the Shapley value of $j$ in an SOCG. Regarding \EA, consider the following instance. Let $N=\{1,2,3\}$, and let a coalition be winning exactly when it contains player~$2$ and at least one of players~$1$ and~$3$. Under $(1,2,3)$, the pivotal prefix is $\{1,2\}$, whose two members are critical, so player~$2$ receives $\frac{1}{2}$. If player~$2$ delays to obtain $(1,3,2)$, she is the only critical player of the pivotal prefix $N$ and receives $1$. Since this game is an SOCG, \CES coincides with \CMES on it, and both violate \EA.
\end{proof}

\paragraph{Query Complexity.}
The GM decomposition can have exponentially many components, so computing \CES through an explicit decomposition is expensive. The allocation itself, however, depends on only a few values of $v$ at each arrival. \citet[Algorithm~2 and Theorem~7]{ZGT+26a} observe this for \eRFC, and the same argument applies to \CES. Consider the arrival of player $i$, and let $t=|N_{\pi\mid i}|$. Since the values of prefixes do not decrease along $\pi$, the components that become winning at this arrival are exactly those whose thresholds $z_k$ lie in the range $(v(N_{\pi\mid i}\setminus\{i\}),\,v(N_{\pi\mid i})]$; their coefficients add up to $\MC_i$, and all other components are paid at other arrivals. As in the proof of \citet[Theorem~7]{ZGT+26a}, a player $j\in N_{\pi\mid i}$ is critical in such a component exactly when $v(N_{\pi\mid i}\setminus\{j\})<z_k$. Each player is thus critical in exactly those new components whose thresholds exceed her deletion value $v(N_{\pi\mid i}\setminus\{j\})$, and the entrant is critical in all of them. The $t$ deletion values therefore cut the range into at most $t$ consecutive intervals, and within each interval all components have the same critical players. Moreover, the coefficients of the components with thresholds in an interval add up to the length of the interval, because its endpoints are values of $v$. The thresholds inside an interval are therefore irrelevant: \eRFC gives each interval to its earliest-arriving critical player, and \CES divides it equally among its critical players. \Cref{alg:ces-direct} computes the \CES payments in this way from $v(N_{\pi\mid i})$ and the $t$ deletion values.

For example, consider the game of \Cref{fig:example} under the order $(d,1,2,4,3)$. The arrival of player~$3$ raises the value from $0$ to $2$. The deletion values are $0$ for player~$3$, $1$ for player~$4$, and $2$ for players $d$, $1$, and $2$. Hence the interval $(0,1]$ has the single critical player~$3$, and the interval $(1,2]$ has the critical players $3$ and $4$. \CES pays $1+\frac{1}{2}=\frac{3}{2}$ to player~$3$ and $\frac{1}{2}$ to player~$4$, whereas \eRFC gives the interval $(1,2]$ to the earlier player~$4$ and pays $1$ to each of them.

\begin{algorithm}[!htbp]
\caption{Computation of \CES rule for prefix subgame $G^i$}
\label{alg:ces-direct}
\begin{algorithmic}[1]
\REQUIRE Current prefix $N_{\pi\mid i}$ with $t=|N_{\pi\mid i}|$, arriving player $i$, and value oracle $v$.
\ENSURE Payment $x_j$ to each player $j\in N_{\pi\mid i}$ at the arrival of $i$.
\IF{$\MC_i=0$}
\RETURN $x_j=0$ for all $j\in N_{\pi\mid i}$
\ENDIF
\FORALL{$j\in N_{\pi\mid i}$}
\STATE $d_j\leftarrow\max\{v(N_{\pi\mid i}\setminus\{i\}),\,v(N_{\pi\mid i}\setminus\{j\})\}$
\ENDFOR
\STATE Sort the players of $N_{\pi\mid i}$ as $\sigma(1),\ldots,\sigma(t)$ so that $d_{\sigma(1)}\leq\cdots\leq d_{\sigma(t)}$
\STATE $b_r\leftarrow d_{\sigma(r)}$ for $1\leq r\leq t$, and $b_{t+1}\leftarrow v(N_{\pi\mid i})$
\STATE $L\leftarrow 0$
\FOR{$r=t,t-1,\ldots,1$}
\STATE $L\leftarrow L+(b_{r+1}-b_r)/r$
\STATE $x_{\sigma(r)}\leftarrow L$
\ENDFOR
\RETURN $(x_j)_{j\in N_{\pi\mid i}}$
\end{algorithmic}
\end{algorithm}

\begin{proposition}\label{prop:ces-poly}
\CES can be computed with $O(n^2)$ value queries.
\end{proposition}

\begin{proof}
Fix an arrival of player $i$. If $\MC_i=0$, no component becomes winning and nothing is paid. Otherwise, let $d_j$, $\sigma$, and $b_1\leq\cdots\leq b_{t+1}$ be as in \Cref{alg:ces-direct}. The entrant has $d_i=v(N_{\pi\mid i}\setminus\{i\})$, and every other $d_j$ is at least this value, so $b_1=v(N_{\pi\mid i}\setminus\{i\})$, and the intervals $(b_r,b_{r+1}]$ with $1\leq r\leq t$ partition the range of new thresholds. Let $z_k$ be a threshold in such an interval. Since $z_k>v(N_{\pi\mid i}\setminus\{i\})$, a player $j$ is critical in the $k$-th component exactly when $d_j<z_k$, that is, exactly when $j\in\{\sigma(1),\ldots,\sigma(r)\}$, and \CMES gives each of these $r$ players the amount $c_k/r$ from this component. Since $b_r$ and $b_{r+1}$ are values of $v$, the coefficients of the components with thresholds in $(b_r,b_{r+1}]$ add up to $b_{r+1}-b_r$, so each of $\sigma(1),\ldots,\sigma(r)$ receives $(b_{r+1}-b_r)/r$ from this interval. Player $\sigma(r)$ is critical in the intervals with indices $r,r+1,\ldots,t$ and in no other, so she receives $\sum_{h=r}^{t}\frac{b_{h+1}-b_h}{h}$, which is the amount that \Cref{alg:ces-direct} assigns to her. Ties among the $d_j$ produce empty intervals and do not affect the result. The algorithm queries $v(N_{\pi\mid i})$ and the $t$ deletion values. Therefore, summing over all arrivals gives $O(n^2)$ value queries.
\end{proof}

\section{Extended \ULMES}\label{app:eulmes}
\Cref{ex:ulmes-not-ea} shows that \ULMES fails \EA once the value grows in several increments that are shared by different sets. By delaying her arrival, player~$3$ merges the two increments into a single one, which she shares with player~$4$ alone. On SOCGs, where the value grows only once, this cannot happen, and \ULMES satisfies \EA (\Cref{prop:ulmes_properties}). The GM decomposition reduces a general game to this case. It is convenient to view it as slicing the value into layers: the component with threshold $z_k$ is the layer between the values $z_{k-1}$ and $z_k$, it has thickness $c_k$, and a coalition contains this layer exactly when its value is at least $z_k$. Each layer is an SOCG, so its value arrives in a single increment, and a delay cannot merge it with another layer.

The \emph{extended Upward Lexicographic Marginal Equal Share} (\DULMES) rule applies \ULMES to each layer separately and adds up the results weighted by the thicknesses (\Cref{app:gm}). It is efficient in every prefix and satisfies \STAY, so it is a weighted sharing rule in the sense of \Cref{sec:weighted}. In \Cref{ex:ulmes-not-ea}, the value has two layers, with thresholds $x$ and $y$ and thicknesses $x$ and $y-x$. In the lower layer, the winning coalitions are those containing $\{1,2,3\}$ or $\{3,4\}$; in the upper layer, they are those containing $\{3,4\}$. Under both orders $(1,2,3,4)$ and $(1,2,4,3)$, the deletion scan selects $\{1,2,3\}$ in the lower layer and $\{3,4\}$ in the upper layer; under the second order, it examines the later-arriving player~$4$ before players $1$ and $2$ and removes her from the sharing set of the lower layer. Player~$3$ therefore receives $\frac{x}{3}+\frac{y-x}{2}$ under both orders, and her delay no longer pays.

\begin{theorem}\label{thm:eulmes}
\DULMES satisfies \PART, \SSTAY, \EA, and \OD. 
\end{theorem}

\begin{proof}
By \Cref{prop:ulmes_properties}, \ULMES satisfies \PART, \SSTAY, and \OD on all monotone OCGs, in particular on all SOCGs, and it satisfies \EA on all SOCGs. \Cref{lem:gm-participation} implies that \DULMES preserves \PART, \SSTAY, and \OD, and \Cref{rem:gm-preservation} guarantees \EA, i.e., a delay does not increase a player's payoff in any layer, and hence not in their weighted sum.
\end{proof}

The explicit GM decomposition can have exponentially many layers, but this alone does not make the induced allocation hard to compute: \CES is also defined layer by layer and needs only $O(n^2)$ value queries (\Cref{prop:ces-poly}). The difference lies in how a layer selects its recipients. Under \CES, a player shares a layer exactly when she is critical in it, which depends only on the value of the current prefix without her; the payments are therefore determined by the values of the prefixes and their single-player deletions. Under \DULMES, each layer runs the deletion scan of \ULMES, which removes several players in turn. The sharing set of a layer can therefore depend on the values of coalitions far below the current prefix, and different layers can select different sets. The following lower bound builds a game in which exponentially many layers select pairwise different sharing sets, and in which changing the value of a single such set shifts the boundary between two layers and thereby changes the allocation.

\begin{proposition}
\label{prop:ulmes-query-lower}
For $n\geq3$, there exists an instance such that computing \DULMES requires at least $2^{n-2}-1$ value queries.
\end{proposition}

\begin{proof}
Consider an OCG $G=(N,v,\pi)$ where $N=\{1,\ldots,n\}$ and $\pi=(1,\ldots,n)$. We give every coalition containing $n$ a different value by reading its other members as binary digits:
\begin{align*}
v(S)=
\begin{cases}
  \sum_{r\in S\setminus\{n\}}2^{r-1} & \text{if } n\in S,\\
  0 & \text{otherwise.}
\end{cases}
\end{align*}
This valuation is normalized and monotone and has zero singleton values. It is superadditive, because two disjoint coalitions cannot both contain $n$. For $1\leq q\leq 2^{n-1}-1$, let $S_q$ be the unique coalition with $n\in S_q$ and $v(S_q)=q$.

All value is generated at the arrival of $n$, when the prefix is $N$. Since the coalitions containing $n$ have pairwise different integer values, the GM decomposition consists of one layer of thickness $1$ for each $q$, whose winning coalitions are those with value at least $q$. We claim that the \ULMES sharing set of this layer is $S_q$. The deletion scan keeps $n$, since every coalition without $n$ has value $0$, and then examines $n-1,\ldots,1$. Suppose that it has kept exactly the players of $S_q$ among $r+1,\ldots,n-1$. Removing $r$ leaves the digits of $q$ above $r$ together with all digits below $r$. Since $2^{r-1}$ exceeds the sum of all digits below $r$, the remaining value is at least $q$, and $r$ is removed, exactly when $r\notin S_q$. By induction, the scan returns $S_q$, so \DULMES shares the unit of the layer with threshold $q$ equally among $S_q$.

Now fix an even $q$ with $2\leq q\leq 2^{n-1}-2$, and let $v_q$ agree with $v$ except that $v_q(S_q)=q+\frac{1}{2}$. The proper subsets and supersets of $S_q$ have values at most $q-1$ and at least $q+1$, so $v_q$ is still monotone, and it is superadditive with zero singleton values for the same reason as $v$. The only effect of the change is that the boundary between the layers with thresholds $q$ and $q+1$ moves up by $\frac{1}{2}$: the first layer gains thickness $\frac{1}{2}$ and the second loses it, while their winning coalitions, and hence their sharing sets $S_q$ and $S_{q+1}$, stay the same. Since $q$ is even, $S_{q+1}=S_q\cup\{1\}$, so half a unit moves to the smaller set $S_q$, and the payoff of player~$n$ increases by $\frac{1}{2}(\frac{1}{ |S_q|}-\frac{1}{ |S_q|+1})>0$. A deterministic algorithm that computes \DULMES exactly and never queries $S_q$ receives the same answers on $v$ and $v_q$ and gives the same output, although the correct allocations differ. Hence, on input $v$, the algorithm must query $S_q$ for every even $q$ with $2\leq q\leq 2^{n-1}-2$. These are $2^{n-2}-1$ distinct coalitions, so the algorithm makes at least $2^{n-2}-1$ value queries.
\end{proof}

The same perturbation leaves the \CES allocation unchanged unless $S_q=N\setminus\{j\}$ for some player $j$. Under \CES, the recipients of the layers with thresholds $q$ and $q+1$ are the players $j$ with $v(N\setminus\{j\})<q$ and with $v(N\setminus\{j\})<q+1$, respectively. Since all values of $v$ are integers, these two sets coincide unless $v(N\setminus\{j\})=q$ for some $j$, and moving half a unit between two layers with the same recipients does not change any payment.

\begin{remark}\label{rem:eulmes-not-ordinal}
\DULMES is not ordinal in the sense of \Cref{def:ordinal}: each layer chooses its recipients by comparisons only, but each layer is weighted by its thickness, which reflects the gaps between the values of $v$. Let $N=\{1,2\}$, $\pi=(2,1)$, $v(\{2\})=0$, $v(\{1\})=a>0$, and $v(N)=b>a$. At the arrival of player~$1$, the layers with thresholds up to $a$ select only player~$1$, and the layers with thresholds in $(a,b]$ select both players, so the weights are $\lambda_1(G)=\frac{a+b}{2b}$ and $\lambda_2(G)=\frac{b-a}{2b}$. Changing $(a,b)$ from $(1,2)$ to $(1,3)$ preserves every comparison between coalition values but changes the weights from $(\frac{3}{4},\frac{1}{4})$ to $(\frac{2}{3},\frac{1}{3} )$.
\end{remark}

\section{Individual Rationality}\label{app:ir}
We finally discuss individual rationality, which requires that each player receive at least her stand-alone value. We begin with a negative result: in general, no sharing rule satisfies \IR.
\begin{proposition}\label{prop:ir-impossible}
No online value-sharing rule satisfies \IR on all OCGs.
\end{proposition}

\begin{proof}
Let $N=\{1,2\}$ and $v(\{1\})=v(\{2\})=v(N)=1$. \IR requires $\phi(G,1)+\phi(G,2)\geq 2>v(N)$, which contradicts efficiency. More generally, an efficient final allocation satisfying \IR exists if and only if $v(N)\geq\sum_{i\in N}v(\{i\})$. 
\end{proof}

We next identify the valuations for which some online value-sharing rule satisfies \STAY and \IR at every stage, under every arrival order. Superadditivity suffices, and the exact condition is the classical zero-monotonicity~\citep{PeSu07a,AdKr17a}, also called weak superadditivity. Zero-monotonicity says that adding a player to a coalition increases its value by at least her stand-alone value. Its relation to the individual rationality of marginal allocations is standard~\citep[Lemma~7]{AdKr17a}.

\begin{definition}[Zero-monotonicity~\citep{PeSu07a}]\label{def:zero-monotone}
A valuation function $v$ is \emph{zero-monotone} if $v(S\cup\{i\})\geq v(S)+ v(\{i\})$ for all $S\subseteq N$ and $i\notin S$. 
\end{definition}

Since every prefix subgame is itself an OCG, a rule that satisfies \IR on a class of games closed under taking prefix subgames guarantees each arrived player her stand-alone value at every stage, not only in the final allocation. The following proposition shows that this requirement, together with \STAY, can be met under every arrival order exactly on zero-monotone valuations.

\begin{proposition}
\label{prop:ir-domain}
For any valuation function $v$, there exists an online value-sharing rule that, under every arrival order $\pi\in\Pi(N)$, satisfies \STAY and satisfies \IR in every prefix subgame of $(N,v,\pi)$ if and only if $v$ is zero-monotone.
\end{proposition}

\begin{proof}
Let $\phi$ denote such an online value-sharing rule. Consider any $S\subseteq N$ and $i\notin S$, and let $G=(N,v,\pi)$ for an arrival order $\pi$ in which the players of $S$ arrive first, followed by $i$. Since $\phi$ is efficient in every prefix and satisfies \STAY, the value $v(S)$ has been distributed among the players in $S$ and is not withdrawn when $i$ arrives. Note that \IR in the prefix subgame $G^{S\cup\{i\}}$ requires at least $v(\{i\})$ for the newly arrived player $i$. Hence $v(S\cup\{i\})\geq v(S)+v(\{i\})$. Since $S$ and $i$ are arbitrary, the condition of zero-monotonicity follows. Conversely, suppose $v$ is zero-monotone. Such an online value-sharing rule exists. For example, under \DMC and any arrival order, player $i$ receives $v(N_{\pi\mid i})-v(N_{\pi\mid i}\setminus\{i\})\geq v(\{i\})$ at her arrival and keeps this payment, so \DMC satisfies \STAY and \IR in every prefix subgame. This concludes the proof.
\end{proof}

Every superadditive valuation is zero-monotone, but the converse fails: on four players, the valuation with $v(S)=|S|+1$ if $|S|\geq 2$ and $v(S)=|S|$ otherwise is zero-monotone, but two disjoint pairs have total value $6>5=v(N)$. In the rest of this appendix, we work on the zero-monotone domain.

\begin{remark}\label{rem:classical-ir}
On zero-monotone games, \DMC and \SV satisfy \IR. This follows from the classical individual rationality of marginal vectors and their convex combinations~\citep[Lemma~7]{AdKr17a}: every marginal contribution of a player is at least her singleton value, and so is their Shapley average.
\end{remark}

\begin{proposition}\label{RFC_is_not_IR}
\eRFC fails \IR, even on superadditive OCGs.
\end{proposition}

\begin{proof}
Let $N=\{1,2\}$, $\pi=(1,2)$, $v(\{1\})=1$, $v(\{2\})=2$, and $v(N)=5$; this game is superadditive. Its GM decomposition has three components, with thresholds $1$, $2$, and $5$ and coefficients $1$, $1$, and $3$, shown in \Cref{example_game_eRFC_not_IR}. In the first component, player~$1$ is pivotal at her own arrival and receives its unit. In the second, the pivotal prefix is $N$, whose only critical player is $2$. In the third, both players are critical in $N$, and \RFC gives the unit to player~$1$, who arrives first. Hence $\phi(G,1)=1+3=4$ and $\phi(G,2)=1<v(\{2\})$.
\end{proof}

\begin{table}[!ht]
\centering
\begin{tabular}{ccccc}
\toprule
 & Coefficient & $\{1\}$ & $\{2\}$ & $\{1,2\}$ \\ \midrule
$v$ & & $1$ & $2$ & $5$ \\ \midrule
$\bar v_1$ & $c_1=1$ & $1$ & $1$ & $1$ \\
$\bar v_2$ & $c_2=1$ & $0$ & $1$ & $1$ \\
$\bar v_3$ & $c_3=3$ & $0$ & $0$ & $1$ \\ \bottomrule
\end{tabular}
\caption{The valuation and the GM decomposition in the proof of \Cref{RFC_is_not_IR}.}
\label{example_game_eRFC_not_IR}
\end{table}

The rules introduced in this paper fail \IR as well. Let $N=\{1,2\}$, $v(\{1\})=1$, $v(\{2\})=2$, $v(N)=3$, and $\pi=(1,2)$. Both players are critical in $N$, and the \ULMES sharing set at the arrival of player~$2$ is $N$. Player~$2$ therefore receives $1$ under \MES, \NDMES, \CMES, and \ULMES, at most $1$ under every \WULMS rule satisfying the four axioms, whose entrant fraction at $t=s=2$ is at most $\frac{1}{2}$ by condition~(a) of \Cref{thm:wulms}, and $\frac{3}{2}$ under \DULMES and \CES. All of these payments are below $v(\{2\})=2$.

We next introduce a simple and natural individual rationality refinement framework. Before that, we first show that \IR and \SSTAY are incompatible, even on additive OCGs.

\begin{proposition}\label{prop:ir-sstay}
No online value-sharing rule satisfies both \IR and \SSTAY on all additive OCGs.
\end{proposition}
\begin{proof}
Let $N=\{1,2\}$, $v(\{1\})=v(\{2\})=1$, $v(N)=2$, and $\pi=(1,2)$. Efficiency gives $\phi(G^1,1)=1$, and \IR together with efficiency forces $\phi(G,1)=\phi(G,2)=1$. Player~$2$ is incremental and $v(N)>v(\{2\})$, so \SSTAY requires $\phi(G^2,1)>\phi(G^1,1)=1$, a contradiction.
\end{proof}

The obstruction is that the new value can be absorbed entirely by stand-alone values. We hence say that a prefix subgame has strict surplus if the marginal contribution of every critical player strictly exceeds her stand-alone value.

\begin{definition}\label{def:strict-surplus}
For any OCG $G=(N,v,\pi)$ and any prefix subgame $G^i$, $G^i$ has \emph{strict surplus} if $v(N_{\pi\mid i})-v(N_{\pi\mid i}\setminus\{j\})>v(\{j\})$ for every $j\in\CR(N_{\pi \mid i})$.
\end{definition}

The following refinement applies to every rule. For a zero-monotone OCG $G=(N,v,\pi)$, let $w(S)=v(S)-\sum_{j\in S}v(\{j\})$ be the \emph{surplus game}, and let $G_w=(N,w,\pi)$. By zero-monotonicity, $w$ is normalized and monotone, so $G_w$ is an OCG under the scope of our settings.

\begin{definition}[\IR Refinement]\label{def:ir-refinement}
For an online value-sharing rule $\phi$, its \emph{\IR refinement} $\phi^{\IR}$ is given by $\phi^{\IR}(G^k,j)=v(\{j\})+\phi(G_w^k,j)$ for every prefix subgame $G^k$ and every $j\in N_{\pi\mid k}$.
\end{definition}

Every arriving player first receives her stand-alone value, and the remaining value is shared by applying $\phi$ to the surplus game.

\begin{proposition}\label{prop:maintain_IR_property}\label{maintain_IR_property}
Let $\phi$ be an online value-sharing rule. On zero-monotone OCGs, $\phi^{\IR}$ is efficient in every prefix and satisfies \IR in every prefix subgame. If $\phi$ satisfies \STAY, \EA, \PART, \OD, or \SF on all OCGs, then so does $\phi^{\IR}$ on zero-monotone OCGs, and if $\phi$ satisfies \SSTAY, then $\phi^{\IR}$ satisfies the requirement of \SSTAY at every prefix with strict surplus.
\end{proposition}

\begin{proof}
Zero-monotonicity gives $w(S\cup\{k\})-w(S)=v(S\cup\{k\})-v(S)-v(\{k\})\geq 0$ for $k\notin S$, so $w$ is normalized and monotone and $G_w$ is an OCG. Hence $\phi(G_w^k,j)\geq 0$, which gives \IR, and
\begin{align*}
	\sum_{j\in N_{\pi\mid k}}\phi^{\IR}(G^k,j)=\sum_{j\in N_{\pi\mid k}}v(\{j\})+w(N_{\pi\mid k})=v(N_{\pi\mid k}),
\end{align*}
which shows that $\phi^{\IR}$ is efficient in every prefix. For \STAY, the term $v(\{j\})$ is constant once $j$ has arrived. For \EA, the term $v(\{i\})$ does not depend on $\pi$, and $\phi$ satisfies \EA on $G_w$. For \PART, let $i$ be incremental and let $T=N_{\pi\mid i}$. If $v(\{i\})>0$, then $\phi^{\IR}(G^i,i)\geq v(\{i\})>0$. If $v(\{i\})=0$, then $w(T)-w(T\setminus\{i\})=v(T)-v(T\setminus\{i\})>0$, so $i$ is incremental in $G_w$ and $\phi(G_w^i,i)>0$. For \OD, a dummy player $j$ of $G^k$ satisfies $v(\{j\})=v(\emptyset)=0$ and $w(S\cup\{j\})-w(S)=v(S\cup\{j\})-v(S)=0$, so $j$ is a dummy player of $G_w^k$ and receives nothing from $\phi$. For \SSTAY, let $i$ be incremental at a prefix $T=N_{\pi\mid i}$ with strict surplus. Every $j\in\CR(T)$ satisfies $w(T)-w(T\setminus\{j\})=v(T)-v(T\setminus\{j\})-v(\{j\})>0$, so $i$ is incremental in $G_w$ and every $j\in\CR(T)\setminus\{i\}$ is critical in $G_w$. \SSTAY of $\phi$ on $G_w$ gives
\begin{align*}
	\phi^{\IR}(G^i,j)-\phi^{\IR}(G^{i-},j)=\phi(G_w^i,j)-\phi(G_w^{i-},j)> 0.
\end{align*}
For \SF, the linearity of the Shapley value gives
\begin{align*}
	\frac{1}{n!}\sum_{\pi\in\Pi(N)}\phi^{\IR}\left((N,v,\pi),j\right)=v(\{j\})+\ShVa(G_w,j)=\ShVa(G,j).
\end{align*}
This completes the proof.
\end{proof}

We also remark that the \IR refinement adds at most $n$ value queries. To compute $\phi^{\IR}$, we query the singleton value $v(\{j\})$ of each player at her arrival; afterwards, every value query to the surplus game $w$ can be answered with a single value query to $v$. Hence, if $\phi$ can be computed with $q(n)$ value queries on all monotone OCGs, then $\phi^{\IR}$ can be computed with at most $q(n)+n$ value queries on zero-monotone OCGs, which is $O(q(n))$ whenever $q(n)=\Omega(n)$.

\end{document}